\documentclass[journal]{IEEEtran}

\usepackage{bbm}
\usepackage{color}
\usepackage{mathrsfs}
\usepackage{pifont}
\usepackage{amsfonts}
\usepackage{citesort}
\usepackage{graphicx}

\ifCLASSINFOpdf
\else
\fi

\usepackage[cmex10]{amsmath}

\newtheorem{theorem}{Theorem}
\newtheorem{definition}{Definition}
\newtheorem{algorithm}{Algorithm}
\newtheorem{example}{Example}

\newtheorem{assumption}{Assumption}
\newtheorem{proposition}{Proposition}

\begin{document}
%
\title{Upper Bounds on the Capacities of Noncontrollable Finite-State
Channels with/without Feedback}

%
%
%

\author{Xiujie~Huang,~\IEEEmembership{Student~Member,~IEEE,}
        Aleksandar~Kav\v{c}i\'{c},~\IEEEmembership{Senior~Member,~IEEE,}
        and~Xiao~Ma,~\IEEEmembership{Member,~IEEE}
\thanks{This work was supported by International Program of Project 985, Sun Yat-sen University,
and by the National Basic Research Program of China (973 Program, No.~2012CB316100),
and by the NSFC~(No.~61172082) and the NSFC and Guangdong Province~(No.~U0635003).
This work was also supported by the NSF, grant CCF10-18984.
This work was performed while X.~Huang was visiting University of Hawaii.}%
\thanks{X.~Huang and X.~Ma are with the Department of Electronic and Communication Engineering,
Sun Yat-sen University, Guangzhou, GD 510006 China~(email:~huangxj5@mail2.sysu.edu.cn,
maxiao@mail.sysu.edu.cn).}
\thanks{A.~Kav\v{c}i\'{c} is with the Department of
Electrical Engineering, University of Hawaii, Honolulu,
HI 96822 USA.}%
\thanks{Manuscript received March 27, 2009; first revision May 22, 2011; accepted February 24, 2012; final version April 7, 2012.}
}

\markboth{Appears in IEEE TRANSACTIONS ON INFORMATION THEORY,
August, 2012}{Huang \MakeLowercase{\textit{et~al.}}: Upper Bounds on the Capacities of Non-Controllable Finite-State
Channels with/without Feedback}




\maketitle\thispagestyle{empty}

\begin{abstract}
Noncontrollable finite-state channels~(FSCs) are FSCs in which the channel inputs have no influence on the channel states, i.e., the channel states evolve freely. Since single-letter formulae for the channel capacities are rarely available for general noncontrollable FSCs, computable bounds are usually utilized to numerically bound the capacities. In this paper, we take the delayed channel state as part of the channel input and then define the {\em directed information rate} from the new channel input~(including the source and the delayed channel state) sequence to the channel output sequence. With this technique, we derive a series of upper bounds on the capacities of noncontrollable FSCs with/without feedback. These upper bounds can be achieved by conditional Markov sources and computed by solving an average reward per stage stochastic control problem~(ARSCP) with a compact state space and a compact action space. By showing that the ARSCP has a uniformly continuous reward function, we transform the original ARSCP into a finite-state and finite-action ARSCP that can be solved by a value iteration method. Under a mild assumption, the value iteration algorithm is convergent and delivers a near-optimal stationary policy and a numerical upper bound.
\end{abstract}



%
\IEEEpeerreviewmaketitle

\section{Introduction}\label{sec1}

\IEEEPARstart{T}{he} channel capacity is usually defined as an operational quantity, called {\em operational capacity}, that is the supremum of all {\em achievable rates}. For a stationary memoryless channel without feedback, it is well-known that the operational capacity equals the maximum mutual information between the channel input and the channel output, called {\em information capacity}~\cite{Shannon48,Cover91}. It is also well-known that feedback does not increase capacities of memoryless
channels~\cite{Cover91,Shannon56}. That is, the {\em feedback capacity} of a memoryless channel also equals the maximum mutual information. However, for a channel with memory, Massey~\cite{Massey90} proved that the feedback capacity is upper-bounded by the normalized {\em directed information}\footnote{Directed information was introduced by Massey~\cite{Massey90} who attributes
it to Marko~\cite{Marko73}. Recently, Venkataramanan and Pradhan~\cite{Venkataramanan05} gave a new interpretation of the directed information.}, which can be strictly less than the mutual information. Since the mutual information can be reduced to the directed information when the channel is used without feedback~\cite{Massey90}, both the feedforward capacity for {\em information stable} channels~\cite{Dobrushin63} and the feedback capacity for {\em directed information stable} channels~\cite{Tatikonda09} can definitely be characterized by a unified quantity, i.e., the limit of the normalized directed information. This fact will be employed in this paper to upper-bound the feedforward/feedback capacities. Although the capacities for general channels can be characterized either by the supremum of the {\em spectral inf-mutual information rates}~\cite{Verdu94,Han03} or by the supremum of the {\em spectral inf-directed information rates}~\cite{Tatikonda09}, they are usually difficult to compute numerically.

In this paper, we are concerned with stationary finite-state channels~(FSCs) as defined in~\cite[p.~97]{Gallager68}, a class of (directed)~information stable channels with memory. Finite-state channels model a class of channels with memory which have finite channel states, such as finite-length intersymbol interference~(ISI) channels and Gilbert-Elliott~(GE) channels~\cite{Mushkin89}. Gallager~\cite{Gallager68} defined the {\em lower capacity} and the {\em upper capacity} to characterize the dependence of the feedforward capacity on the initial channel state and showed that they coincide for {\em indecomposable} FSCs. Permuter {\em et~al.}~\cite{Permuter09} extended Gallager's method to characterize the feedback capacity of FSCs. For a class of stationary FSCs with feedback~\cite{Kim08}, Kim proved a coding theorem using an encoding scheme based on block ergodic decomposition and a decoding scheme based on strong typicality. For other special FSCs with/without feedback such as GE channels, GE-like channels and unifilar FSCs, see, for example,~\cite{Mushkin89,Goldsmith96,Permuter08} and the references therein. If the channel state information~(CSI) is known to either one of the transmitter and the receiver or both, the capacity usually has a simplified form. For an example, considering the special class of FSCs without ISI defined in~\cite{Viswanathan99}, if the receiver has perfect CSI and both the output and the channel state are fed back to the transmitter, the feedback capacity can be characterized by a single letter formula.

In addition to the derivation of the capacity formula, the computation of the channel capacity is also an important problem. For general channels, this could be a very complicated optimization problem due to the following two issues. Firstly, the capacity usually takes the form of a limit, whose analytical properties are rarely known. Secondly, it might be required to consider almost all possible input processes to conduct the optimization. A brief review of the computation of the channel capacity or its bounds is summarized as follows.

For the discrete memoryless channel, the capacity can be computed by the Blahut-Arimoto algorithm~\cite{Blahut72,Arimoto72}. For the ISI channel with additive white Gaussian noise, if continuous channel inputs are allowed, the capacity can be computed by using the water-filling theorem~\cite{Gallager68,Cover91}. If only finite channel inputs are allowed in the ISI channel, bounds on the i.u.d. capacity $C_{i.u.d.}$, which is defined as the information rate when the channel inputs are independent and uniformly distributed~(i.u.d.), can be evaluated numerically by a Monte Carlo method~\cite{Hirt88}. A more refined Monte Carlo method that utilizes the BCJR algorithm can be used to numerically evaluate the $C_{i.u.d.}$ and the information rates of stationary FSCs with Markov inputs~\cite{Arnold01,Pfister01,Sharma01,Arnold06}. For an FSC with a given-order Markov input processe, the information rate can be further optimized by a generalization of the Blahut-Arimoto algorithm~\cite{Kavcic01,Vontobel08}. These methods, coupled with the proofs~\cite{Chen08} that Markov processes asymptotically achieve feedforward capacities of ISI channels, can be utilized to very closely lower-bound the feedforward capacities of ISI channels. For upper bounds on the feedforward capacities of the stationary FSCs, see~\cite{Yang05,Huang09} and the references therein.

To compute the feedback capacity of the Markov channel, Tatikonda and Mitter~\cite{Tatikonda00,Tatikonda05,Tatikonda09} introduced a dynamic programming framework based on certain sufficient statistics. However, for general FSCs, the sufficient statistics could be
very complicated and the corresponding dynamic programming problem can not be solved efficiently. Nonetheless, for some special FSCs, efficient dynamic programming algorithms have been implemented to evaluate the feedback capacities numerically~\cite{Yang05,Chen05,Permuter08,Zhao10}.

In this paper, we focus on the stationary {\em noncontrollable} FSC~\cite[Definition~22]{Vontobel08}, which is also known as
Markov channel without ISI~\cite[Definition~6.1]{Tatikonda09}\footnote{The results in this paper can also be applied to {\em hybrid} channels that have both an ISI component and a noncontrollable component.}. By uncontrollability, we mean that the input has no influence on the channel state and the channel state evolves freely. As mentioned previously, for some special noncontrollable FSCs
such as the GE channel~\cite{Mushkin89} and GE-like channels~\cite{Goldsmith96}, the capacity-achieving distributions are known,
and the feedforward capacities can be evaluated using the methods in~\cite{Arnold01,Pfister01,Sharma01,Arnold06}. For general noncontrollable FSCs, however, closely bounding the feedforward capacity and the feedback capacity seems to be the only practical approach. While good lower bounds on the capacities of noncontrollable channels are known~\cite{Vontobel08,Sadeghi09}, computable upper bounds are loose. Here, the main practical result of this paper is the development of a numerical technique to closely upper bound the capacity, which combined with the previously mentioned lower bounds~\cite{Vontobel08,Sadeghi09} delivers a good numerical approximation of the capacity.

The main objective of this paper is to find computable upper bounds on the feedforward and feedback capacities. Firstly and most importantly, we develop upper bounds on the capacities by two techniques. One is inserting the delayed channel state into the channel input and then defining the {\em directed information rate} from the new channel input~(including the source and the delayed channel state) sequence to the channel output sequence. The other is majorizing the set of the considered channel input processes. In this way, we develop two nested sequences of upper bounds for feedforward and feedback capacities, respectively. Secondly, through three theorems, we show that the upper bounds can be achieved by {\em finite-order conditional Markov sources}, conditioned on the delayed feedback~(FB), on the delayed state information~(SI) and on the statistic of channel outputs~(called the {\em a posteriori} probability vector). Thirdly, similar to~\cite{Yang05}, we formulate the computation of the upper bound as an average reward per stage stochastic control problem~(ARSCP) with a continuous state space and a continuous action space~\cite{Bertsekas05,Bertsekas07}. This ARSCP is shown to have a uniformly continuous reward function and can be transformed into a finite-state and finite-action ARSCP, which can be solved by a value iteration method. Under a mild assumption, the value iteration algorithm is convergent and delivers a near-optimal stationary policy as well as a numerical upper bound.

\textbf{Structure:} The rest of this paper is structured as follows. The channel model is given in the next section. In Section~\ref{sec3}, the channel capacities of noncontrollable FSCs with/without feedback are introduced and the upper bounds on the capacities are developed. To facilitate the computation of these bounds, three theorems are presented in Section~\ref{sec4}. In Subsection~\ref{sec5.1}, the computation of upper bounds is formulated as an ARSCP with a compact state space and a compact action space~({\bf Problem~A}) which can be further transformed into a finite-state and finite-action ARSCP~({\bf Problem~B}). In Subsection~\ref{sec5.2}, a value iteration method is introduced to solve {\bf Problem~B} to obtain a near-optimal policy. Section~\ref{sec6} presents some numerical results, followed by the conclusion in Section~\ref{sec7}.

\textbf{Notation:}
A random variable is denoted by an upper-case letter~(e.g. $X$) and its realization is denoted by the corresponding lower-case letter~(e.g. $x$). A vector of random variables $[X_i,X_{i+1},\ldots,X_j]$ is shortly denoted by $X_i^j$ and its realization is denoted by $x_i^j$. By default, we set $X^j\stackrel{\Delta}{=}X_1^j$ and $x^j\stackrel{\Delta}{=}x_1^j$. The cardinality of a set $\mathcal{X}$ is denoted by $\left|\mathcal{X}\right|$. The expectation of a function $g(\cdot)$ of a random variable $X$ is denoted by ${\bf E}[g(X)]$, while the expectation of a function $g(\cdot)$ of a random variable $X$ conditioned on a realization $y$ of a random variable $Y$ is denoted by ${\bf E}_{X|y}[g(X)]$.



\section{Channel Model}\label{sec2}

Let $S_t$, $X_t$ and $Y_t$ denote the channel state, the channel
input and the channel output at time $t\in \mathbb{Z}$, whose
realizations are $s_t$, $x_t$ and $y_t$, respectively. Each state
$s_t$, each input letter $x_t$ and each output letter $y_t$ are
drawn from finite alphabets $\mathcal{S}$, $\mathcal{X}$ and $\mathcal{Y}$, respectively.
More specifically, an FSC has a state sequence $\mathbf{s}=s_0,s_1,s_2,\ldots,s_N$,
an input sequence $\mathbf{x}=x_1,x_2,\ldots,x_N$ and an output sequence $\mathbf{y}=y_1,y_2,\ldots,y_N$.
As in~\cite{Gallager68}, an FSC can be characterized by
\begin{equation}\label{eqnGFSC}
 {\rm Pr}\!\left(y_t,s_t\!\left|x^{t}\!,s_0^{t-1}\!,y^{t-1}\right.\!\right)
    \!=\!{\rm Pr}\!\left(y_t,s_t\!\left|x_t,\!s_{t-1}\right.\!\right).
\end{equation}
An FSC is said to be noncontrollable if the channel inputs have no influence on the channel states and the channel
states evolve freely. Hence, a noncontrollable FSC can further be characterized by
\begin{equation}\label{eqnFSC}
 {\rm Pr}\!\left(y_t,s_t\!\left|x^{t}\!,s_0^{t-1}\!,y^{t-1}\right.\!\right)
    \!=\!{\rm Pr}\!\left(y_t\!\left|x_t,\!s_{t-1}\right.\!\right){\rm Pr}\!\left(s_t\!\left|s_{t-1}\right.\!\right).
\end{equation}
Moreover, we assume that the noncontrollable FSC is stationary and indecomposable~\cite{Gallager68}, that is,
the right-hand side of~(\ref{eqnFSC}) is independent of time $t$ and the effect of the initial state $s_0$ on the characteristic of the channel dies away with time. For this reason, without loss of generality, we make an assumption that the distribution of the initial state $S_0$ equals the stationary distribution of the state $S_t$ where $t \geq 1$.


{\bf Remark:} Under the above assumptions, it is easy to verify that if there is no feedback, then given the channel state
$s_{t-1}$ and channel input $x_t$, the channel output $y_t$ and state $s_t$ are
statistically independent of other channel inputs and prior channel states and outputs, i.e., for $t\leq N$,
\begin{equation}\label{eqnNOfb}
    {\rm Pr}\!\left(y_t,s_t\!\left|x^{N},s_0^{t-1},y^{t-1}\right.\!\right)
    \!=\!{\rm Pr}\!\left(y_t\!\left|x_t,\!s_{t-1}\right.\!\right){\rm Pr}\!\left(s_t\!\left|\!s_{t-1}\right.\!\right).
\end{equation}
However, if feedback is allowed~(precisely speaking, the output sequence $y^{t-1}$ is available at
the transmitter before emitting symbol $X_t$), then equality~(\ref{eqnNOfb}) may not hold.

The noncontrollable FSC will be illustrated by the following
example related to the Gilbert-Elliott~(GE) channel.
\begin{example}[The RLL$(1,\infty)$-GE Channel]\label{GE}
  The channel input is required to be a binary run-length-limited~(RLL) sequence satisfying the RLL$(1,\infty)$
  constraint, i.e., there are no consecutive ones in the
  sequence~(see Fig. \ref{RLL}). The channel is a GE channel with two
  states~(see Fig. \ref{GEg}), a ``good'' state and a ``bad'' state. Denote the channel state alphabet
  by $\mathcal{S}\stackrel{\Delta}{=}\{g,b\}$. The transition probabilities
  between channel states are $p(b|g)\stackrel{\Delta}{=}{\rm Pr}\left(S_t=b\left|S_{t-1}=g\right.\right)$ and
  $p(g|b)\stackrel{\Delta}{=}{\rm Pr}\left(S_t=g\left|S_{t-1}=b\right.\right)$.
  When the channel state is ``good'', i.e., $S_{t-1}=g$,
  the channel acts as a binary symmetric channel~(BSC) with
  cross-over probability $\varepsilon_g$. When the channel is
  ``bad'', i.e., $S_{t-1}=b$, the channel is a BSC with cross-over probability
  $\varepsilon_b$.
\hfill \ding{113}
\end{example}

\begin{figure}[!t]
\centering
\includegraphics[width=1.0in]{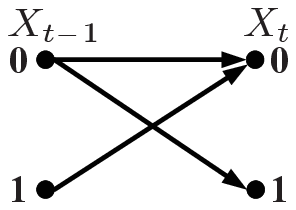}
\caption{A trellis section of the RLL$(1,\infty)$ sequence.}
\label{RLL}
\end{figure}
\begin{figure}[!t]
\centering
\includegraphics[width=2.35in]{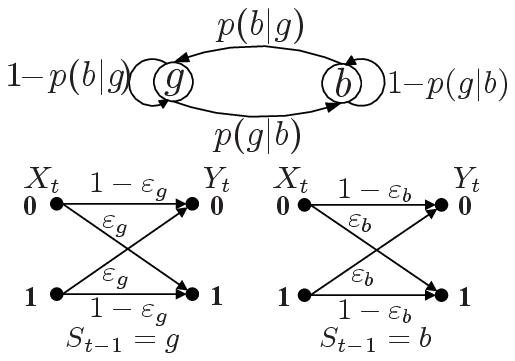}
\caption{A Gilbert-Elliott channel.} \label{GEg}
\end{figure}

\section{Channel Capacities and Upper bounds}\label{sec3}

\subsection{Channel Capacities}\label{subsecCapacity}

In order to unify the presentations of both channel capacities~(the
feedforward capacity and the feedback capacity), we use the notion
of directed information, which was introduced by Massey in~\cite{Massey90}.
For any given joint probability distribution ${\rm Pr}\left(x^N,y^N\right)$,
the directed information from the channel input sequence $X^N$ to channel
output sequence $Y^N$ is defined as
\begin{equation*}
 I\left(X^N\rightarrow Y^N\right) \stackrel{\Delta}{=}
 \sum_{t=1}^N\,I\!\left(X^{t};Y_t\left|Y^{t-1}\right.\!\right).
\end{equation*}
It has been shown that $I\left(X^N\rightarrow Y^N\right)\leq I\left(X^N;Y^N\right)$
with equality if the channel is used without feedback~\cite{Massey90}.
For simplicity, we denote $\mathcal{I}\left(X\rightarrow Y\right)$ as the directed
information rate from the channel input to the channel output, that is,
\begin{equation}\label{drate}
 \mathcal{I}\left(X\rightarrow Y\right) \stackrel{\Delta}{=}
 \liminf_{N\rightarrow\infty}\frac{1}{N}\,I\!\left(X^N\rightarrow Y^N\right).
\end{equation}

We now prove that the capacities can be characterized by the suprema of the directed information rates.

\begin{theorem}\label{ThemCap}
The feedforward capacity of a stationary indecomposable noncontrollable FSC
is given by
\begin{equation}\label{eqnCap}
  C = \sup_{\left\{{\rm Pr}\left(x_t\left|x^{t-1}\right.\!\right)\right\}_{t=1}^{\infty}}
  \mathcal{I}\left(X \rightarrow Y\right)
\end{equation}
where the supremum is taken over all possible channel input processes.
The feedback capacity of a stationary indecomposable noncontrollable FSC is given by
\begin{equation}\label{eqnCapFB}
  C^{fb} = \sup_{\left\{{\rm Pr}\left(x_t\left|x^{t-1},y^{t-1}\right.\!\right)\right\}_{t=1}^{\infty}}
  \mathcal{I}(X\rightarrow Y)
\end{equation}
where the supremum is taken over all possible channel input processes
that are causally dependent on the past channel outputs. This means
that all past channel outputs $Y^{t-1}$ must be fed back to the source
before emitting the symbol $X_t$.
\end{theorem}

\begin{IEEEproof}
See Appendix~\ref{proofThemCap}.
\end{IEEEproof}

For the general FSC, based on certain sufficient statistics, a dynamic programming framework to evaluate the capacity was presented~\cite{Tatikonda09}. However, as mentioned in Section~VIII of~\cite{Tatikonda09}, the sufficient statistic for a general FSC is often too complicated to be employed in dynamic programming methods. For some {\em special} FSCs, efficient dynamic programming algorithms have been proposed to evaluate the feedback capacity numerically~\cite{Permuter08,Yang05,Chen05,Zhao10}.
The main objective of this paper is to develop numerically computable upper bounds on the capacities of {\em general} indecomposable noncontrollable FSCs~(\ref{eqnFSC}) with/without feedback.

\subsection{Upper Bounds on Capacities}

To upper-bound the capacities, a technique of inserting the delayed channel state into the channel input is employed.
Then the directed information from the channel input and delayed channel state sequence to the channel output sequence
can be well defined as follows.

\begin{definition}\label{defDINRupb}
For a stationary indecomposable noncontrollable FSC~(\ref{eqnFSC}), the {\em directed information rate}
$\mathcal{I}_v\!\left(X,S\rightarrow Y\right)$ is defined as
  \begin{equation}\label{newrate}
     \mathcal{I}_v\left(X,S\rightarrow Y\right)
     \stackrel{\Delta}{=}\liminf_{N\rightarrow\infty}\frac{1}{N}\sum_{t=1}^N I\!\left(X^{t},S_0^{t-v-1};Y_t\left|Y^{t-1}\right.\!\right).
  \end{equation}
\hfill  \ding {113}
\end{definition}

In this definition, the $v$-delayed channel state is considered as a part of the channel input. Obviously, for a given channel input process, there is a nested sequence of upper bounds on $\mathcal{I}\left(X\rightarrow Y\right)$ as
\begin{eqnarray}\label{dupbound}
 \lefteqn{\mathcal{I}\left(X\rightarrow Y\right)\leq \cdots \leq \mathcal{I}_{v+1}\left(X,S\rightarrow Y\right)} \nonumber\\
      & & \;\;\;\;\;\;\;\;\;\:\; \leq \mathcal{I}_v\left(X,S\rightarrow Y\right) \leq \cdots \leq \mathcal{I}_{0}\left(X,S\rightarrow Y\right)\!.
\end{eqnarray}
Furthermore, the capacities in Theorem~\ref{ThemCap} can be bounded as
\begin{equation}\label{eqnCupbsim}
\begin{array}{ccc}
  C &\leq& \sup\limits_{\left\{\text{Pr}\left(x_t|x^{t-1}\right)\right\}_{t=1}^{\infty}}\mathcal{I}_v\left(X,S\rightarrow Y\right) \\
  C^{fb} &\leq& \sup\limits_{\left\{\text{Pr}\left(x_t|x^{t-1},y^{t-1}\right)\right\}_{t=1}^{\infty}}\mathcal{I}_v\left(X,S\rightarrow Y\right).
\end{array}
\end{equation}
These upper bounds, however, can not be easily evaluated because
the source sets are too general to be specified with a few
parameters. To develop simpler expressions for upper bounds, we
need to define the following sources in a similar way to those
in~\cite{Huang09}.

\begin{definition}\label{defuFBvSI}
Assume that the \textbf{$\emph{u}$-delayed output
feedback~(FB)} $Y^{t-u-1}$, and the
\textbf{$\emph{v}$-delayed state information~(SI)}
$S_0^{t-v-1}$ are available at the source just before the emission
of $X_t$~(see Fig.~\ref{figFBFSC}). Then the channel input $X_t$
could be selected according to a preset conditional probability
law
${\rm Pr}\!\left(x_t\!\left|x^{t-1},s_0^{t-v-1},y^{t-u-1}\right.\!\right)$.
All such input processes $\left\{X_t\right\}$ are described by a
set $\mathcal{P}(u,v)$, i.e.,
  \begin{equation*}\label{psource}
  \mathcal{P}(u,v) \stackrel{\Delta}{=} \left\{{\rm Pr}\!\left(x_t\!\left|x^{t-1},s_0^{t-v-1},y^{t-u-1}\right.\!\right)\right\}_{t=1}^{\infty}.
  \end{equation*}
In other words, $\mathcal{P}(u,v)$ represents the set of all sources~(channel inputs) with $u$-delayed FB and $v$-delayed SI.
\hfill \ding {113}
\end{definition}
\begin{figure}[!t]
\centering
\includegraphics[width=2.5in]{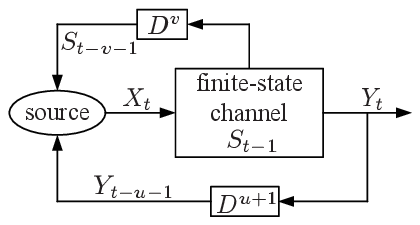}
\caption{A noncontrollable FSC model with $u$-delayed FB and
$v$-delayed SI.} \label{figFBFSC}
\end{figure}

Note that the delays $u$ and $v$ are both non-negative.
An important subclass of sources from $\mathcal{P}(u,v)$, called conditional
Markov source, is defined as follows.

\begin{definition}\label{defCondMSY}
  For $v \leq m$, a source sequence $\{X_t\}$ used with $u$-delayed FB and $v$-delayed
  SI is said to be \textbf{an $\emph{m}$-th order conditional Markov
  source}
  if the conditional probability mass function satisfies
  \begin{equation*}
   {\rm Pr}\!\left(x_t\!\left|x^{t-1}\!,s_0^{t-v-1}\!,y^{t-u-1}\right.\!\right)
   \!=\!{\rm Pr}\!\left(x_t\!\left|x_{t-m}^{t-1},s_{t-m-1}^{t-v-1},y^{t-u-1}\right.\!\right)\!.
  \end{equation*}
Let $\mathcal{P}_m(u,v)$ represent the set of all such sources,
that is,
  \begin{equation*}\label{msource}
    \mathcal{P}_m(u,v)\stackrel{\Delta}{=}
    \left\{{\rm Pr}\!\left(x_t\left|x_{t-m}^{t-1},s_{t-m-1}^{t-v-1},y^{t-u-1}\right.\!\right)\right\}_{t=1}^{\infty}.
  \end{equation*}
\hfill \ding {113}
\end{definition}

From the definitions of sources $\mathcal{P}(u,v)$ and
$\mathcal{P}_m(u,v)$, we have the following facts for non-negative $u$, $v$ and $m$.
\begin{itemize}
\item The sets of channel input processes $\left\{{\rm Pr} \left(x_t
    \left|x^{t-1}\right.\right)\right\}_{t=1}^{\infty}$ and $\left\{{\rm Pr} \left(x_t
    \left|x^{t-1}, y^{t-1}\right.\right)\right\}_{t=1}^{\infty}$ are subsets of the
    conditional source sets $\mathcal{P}(u,v)$ and $\mathcal{P}(0,v)$,
    respectively.
\item $\mathcal{P}(u+1,v+1)\subseteq \mathcal{P}(u+1,v) \subseteq
    \mathcal{P}(u,v)$ and\\
    $\mathcal{P}(u+1,v+1)\subseteq \mathcal{P}(u,v+1) \subseteq
    \mathcal{P}(u,v)$.
\item If $v+1\leq m$, then\\
    $\mathcal{P}_m(u+1,v+1)\subseteq \mathcal{P}_m(u+1,v) \subseteq \mathcal{P}_m(u,v)$ and\\
    $\mathcal{P}_m(u+1,v+1)\subseteq \mathcal{P}_m(u,v+1) \subseteq \mathcal{P}_m(u,v)$.
\item If $v\leq m$, then
    $\mathcal{P}_{m}(u,v)\!\subseteq\!\mathcal{P}_{m+1}(u,v)\!\subseteq\!\cdots\!\subseteq\!\mathcal{P}(u,v)$.
\end{itemize}
Moreover, we can prove the following proposition.

\begin{proposition}\label{Propnewsource}
For a noncontrollable FSC with sources in the set $\mathcal{P}(u,u)$,
  \begin{equation}\label{eqnPropNewSy}
   {\rm Pr}\!\left(y_t,s_t\!\left|x^{t+u},\!s_0^{t-1},y^{t-1}\right.\!\right)\!=\!
   {\rm Pr}\!\left(y_t\!\left|x_{t},\!s_{t-1}\right.\!\right){\rm Pr}\!\left(s_t\!\left|s_{t-1}\right.\!\right).
  \end{equation}
\end{proposition}

\begin{IEEEproof}
In the case of $u = 0$, equality~(\ref{eqnPropNewSy}) holds from the characteristics of the noncontrollable FSC in~(\ref{eqnFSC}).
In the case of $u \geq 1$, we have
\begin{eqnarray}
  \lefteqn{{\rm Pr}\!\left(y_t,s_t\!\left|x^{t+u}, s_0^{t-1},y^{t-1}\right.\!\right)}\nonumber\\
  &=& \frac{{\rm Pr}\!\left(x^{t+u}, s_0^{t},y^{t}\right)}{{\rm Pr}\!\left(x^{t+u}, s_0^{t-1},y^{t-1} \right)}\nonumber\\
  &=& \frac{{\rm Pr}\!\left(x^{t}, s_0^{t},y^{t}\right)
    {\rm Pr}\!\left(x_{t+1}^{t+u}\!\left|x^t, s_0^{t}, y^t\right.\!\right)}
    {{\rm Pr}\!\left(x^{t}, s_0^{t-1},y^{t-1} \right) {\rm Pr}\!\left(x_{t+1}^{t+u}\!\left|x^{t}, s_0^{t-1},y^{t-1}\right.\! \right)}\nonumber\\
  &\stackrel{\rm (a)}{=}& {\rm Pr}\!\left(y_t,s_t\!\left|x^{t}, s_0^{t-1},y^{t-1}\right.\!\right)\nonumber\\
  &=& {\rm Pr}\!\left(y_t\!\left|x_{t}, s_{t-1}\right.\!\right) {\rm Pr}\!\left(s_t \left|s_{t-1}\right.\!\right)
 \end{eqnarray}
where equality~(a) results from the equality
\begin{equation*}
    {\rm Pr}\! \left(x_{t+1}^{t+u} \!\left|x^t, s_0^{t}, y^t\right. \!\right)
    = {\rm Pr}\! \left(x_{t+1}^{t+u} \!\left|x^{t}, s_0^{t-1},y^{t-1}\right.\! \right)
\end{equation*}
since channel input processes are in the set $\mathcal{P}(u,u) = \left\{{\rm Pr}(x_t|x^{t-1},s_0^{t-u-1},y^{t-u-1})\right\}_{t=1}^{\infty}$~(see Definition~\ref{defuFBvSI}).
\end{IEEEproof}

Proposition~\ref{Propnewsource} implies that the probabilities
${\rm Pr} \left(y_t,s_t \left|x^{t+u}, s_0^{t-1},y^{t-1}\right. \right)$
are unaffected by the source selection from $\mathcal{P}(u,u)$ and that the probabilities
${\rm Pr} \left(y_t,s_t \left|x^{t+u}, s_0^{t-1},y^{t-1}\right. \right)$
can be characterized by the channel only.
From the definition of the set $\mathcal{P}(u,u)$, we directly
introduce a supremum as follows, which will be shown to be an
upper bound on the capacity of the noncontrollable FSC.

\begin{definition}\label{defCapupb}
  Define
  $\mathcal{I}^{*}_{FB,SI}(u,v)$ as the supremum of the information
  rates $\mathcal{I}_v\left(X,S\rightarrow Y\right)$
  over all sources with
  $u$-delayed FB and $u$-delayed SI in $\mathcal{P}(u,u)$, that is,
  \begin{equation}\label{eqnIstar}
    \mathcal{I}^{*}_{FB,SI}(u,v)
     \stackrel{\Delta}{=}\sup\limits_{\mathcal{P}(u,u)}\mathcal{I}_v\left(X,S\rightarrow Y\right).
  \end{equation}
\hfill \ding {113}
\end{definition}

Combining the inequalities in~(\ref{dupbound}) and~(\ref{eqnCupbsim}) with the discussion
after Definitions~\ref{defuFBvSI} and~\ref{defCondMSY}, we conclude this section with the following proposition.

\begin{proposition}\label{PropCapUpB}
\begin{enumerate}
 \item For any $u\geq0$ and $v\geq0$, we have
  \begin{eqnarray*}
   \mathcal{I}^{*}_{FB,SI}\left(u+1,v+1\right) &\leq& \mathcal{I}^{*}_{FB,SI}\left(u+1,v\right) \nonumber\\
   &\leq& \mathcal{I}^{*}_{FB,SI}\left(u,v\right)
  \end{eqnarray*}
  and
  \begin{eqnarray*}
   \mathcal{I}^{*}_{FB,SI}\left(u+1,v+1\right) &\leq& \mathcal{I}^{*}_{FB,SI}\left(u,v+1\right) \nonumber\\
   &\leq& \mathcal{I}^{*}_{FB,SI}\left(u,v\right).
  \end{eqnarray*}

 \item For any $v\geq1$, we have a nested sequence of upper bounds on the feedforward capacity
  \begin{equation*}
   \begin{array}{ccl}
    C \leq \cdots \!\!\!\!&\leq&\!\!\!\! \mathcal{I}^{*}_{FB,SI}\left(v,v\right) \leq \cdots \\
      \!\!\!\!&\leq&\!\!\!\! \mathcal{I}^{*}_{FB,SI}\left(1,1\right)\leq\mathcal{I}^{*}_{FB,SI}\left(0,0\right).
   \end{array}
  \end{equation*}

 \item For any $v\geq1$, we have a nested sequence of upper bounds on the feedback capacity
  \begin{equation*}
   \begin{array}{ccl}
    C^{fb}\leq\cdots\!\!\!\!&\leq&\!\!\!\!\mathcal{I}^{*}_{FB,SI}\left(0,v\right)\leq\cdots \\
    \!\!\!\!&\leq&\!\!\!\! \mathcal{I}^{*}_{FB,SI}\left(0,1\right)\leq\mathcal{I}^{*}_{FB,SI}\left(0,0\right).
   \end{array}
  \end{equation*}
\end{enumerate}
\end{proposition}

\begin{proof}
It is straightforward and omitted here.
\end{proof}

\section{Three Theorems for Upper Bounds}\label{sec4}

In this section, we introduce three main theorems that simplify
the expressions for the upper bounds presented in
Proposition~\ref{PropCapUpB} on the capacities of noncontrollable
FSCs.

%
%
\begin{theorem}\label{ThemshortS}
Let $v\geq 0$. For noncontrollable FSCs,
\begin{equation}
    I(X^{t},S_0^{t-v-1};Y_t|Y^{t-1}) = I(X_{t-v}^{t},S_{t-v-1};Y_t|Y^{t-1})
\end{equation}
and the directed information rate $\mathcal{I}_v\left(X,S\rightarrow
Y\right)$ in (\ref{newrate}) can be simplified as
\begin{equation}\label{rate1}
   \mathcal{I}_v\left(X,S\rightarrow Y\right)
   =\liminf\limits_{N\rightarrow\infty}\frac{1}{N}\sum\limits_{t=1}^NI(X_{t-v}^{t},S_{t-v-1};Y_t|Y^{t-1}).
\end{equation}
\end{theorem}

\begin{IEEEproof}
For any $v\geq 0$, by using the chain rule for mutual information, we
have
\begin{eqnarray}
    \lefteqn{I(X^{t},S_0^{t-v-1};Y_t|Y^{t-1})} \nonumber\\
    &=&\!\!\! I(X_{t-v}^{t},S_{t-v-1};Y_t|Y^{t-1}) \nonumber\\
    &&\!\!\! + \:I(X^{t-v-1},S_0^{t-v-2};Y_t|Y^{t-1},X_{t-v}^{t},S_{t-v-1}).
\end{eqnarray}
The last term equals zero since the current channel output $Y_t$ is independent of the distantly past states $S_0^{t-v-2}$ and inputs $X^{t-v-1}$ if the recent state $S_{t-v-1}$ and inputs $X_{t-v}^t$ and the whole history of outputs $Y^{t-1}$ are given.
\end{IEEEproof}


\begin{theorem}\label{Themnewsource1} Let $0\leq u\leq v$. The supremum $\mathcal{I}^{*}_{FB,SI}(u,v)$
is achieved by a $v$-th order conditional Markov source with
$u$-delayed FB and $u$-delayed SI, that is, \vskip -0.25cm
\begin{equation*}
    \mathcal{I}^{*}_{FB,SI}(u,v)=\sup\limits_{\mathcal{P}_v(u,u)}
     \mathcal{I}_v(X,S\rightarrow Y)
\end{equation*}
where $\mathcal{P}_v(u,u) =
  \left\{{\rm Pr}\!\left(x_t\left|x_{t-v}^{t-1},s_{t-v-1}^{t-u-1},y^{t-u-1}\right.\!\right)\right\}_{t=1}^{\infty}$.
\end{theorem}

\begin{IEEEproof}
See Appendix~\ref{proofThemnewsource1}.
\end{IEEEproof}
%

By Theorem~\ref{Themnewsource1}, to evaluate the supremum
$\mathcal{I}^{*}_{FB,SI}(u,v)$, it is necessary to search the whole
set of conditional probabilities
$\left\{{\rm Pr}\!\left(x_t\left|x_{t-v}^{t-1},s_{t-v-1}^{t-u-1},y^{t-u-1}\right.\!\right),
t=1,2,\ldots,\right\}$. As time $t$ increases, the space of
sequences $y^{t-u-1}$ expands exponentially, which makes it
complicated to keep track of the dependence of the process $X_t$
on $Y^{t-u-1}$. In the sequel, we find some finite-size sufficient
statistics to represent the sequence $y^{t-u-1}$.

Let $\mathcal{M}$ be the Cartesian product $\mathcal{X}^v \times
\mathcal{S}^{v-u+1}$ whose elements are indexed simply by $\ell
\in \{0, 1, \cdots, M-1\}$ with $M = |\mathcal{M}|$. A random
vector $\underline{A}_{t}$ is specified as the {\em a posteriori}
probability vector with realization
 \begin{equation}\label{eqnDefAlpha}
    \underline{\alpha}_{t}\stackrel{\Delta}{=}\left[\alpha_{t}(0),\alpha_{t}(1),\cdots,\alpha_{t}(M-1)\right]
 \end{equation}
where
 \begin{equation}\label{eqnDefAlpha1}
  \alpha_{t}(\ell)\!\stackrel{\Delta}{=}\!{\rm Pr}\!\left(\left(\!X_{t-v+1}^{t},\!S_{t-v}^{t-u}\!\right)\!=\!\ell\,|y^{t-u}\!\right)
 \end{equation}
for
$\ell\!\in\!\{0,1,\cdots,M\!-\!1\}$. The sample space of the random vector $\underline{A}_{t}$ is denoted by $\mathcal{A}$,
which is a simplex in $\mathbb{R}^M$. That is, $\mathcal{A}=\{\underline{\alpha}=\left[\alpha(0),\ldots, \alpha(M-1)\right]: \alpha(i) \geq 0, \sum_{i=0}^{M-1}\alpha(i) = 1\}$.
Given the probability vector $\underline{\alpha}_{t-1}$, the channel output $y_{t-u}$ and the set of transition probabilities
${\rm Pr}\!\left(x_t\!\left|x_{t-v}^{t-1},s_{t-v-1}^{t-u-1},y^{t-u-1}\right.\!\right)$,
we can use the forward recursion of the BCJR algorithm~\cite{BCJR74} to compute all values of $\alpha_t(\ell)$ as
\begin{equation}\label{bcjr0}
   \alpha_{t}\!\left(x_{t-v+1}^{t},\!s_{t-v}^{t-u}\right)\!=\!
   \frac{\sum\limits_{x_{t-v},s_{t-v-1}}\!\!\!\!{\rm Pr}\!\left(x_{t-v}^{t},\!s_{t-v-1}^{t-u},\!y_{t-u}\!\left|y^{t-u-1}\right.\!\right)}
      {\sum\limits_{x_{t-v}^{t},s_{t-v-1}^{t-u}}\!\!\!\!{\rm Pr}\!\left(x_{t-v}^{t},\!s_{t-v-1}^{t-u},\!y_{t-u}\!\left|y^{t-u-1}\right.\!\right)}
\end{equation}
where
\begin{eqnarray}\label{bcjr1}
 \lefteqn{{\rm Pr}\!\left(\!x_{t\!-v}^{t},\!s_{t\!-v-\!1}^{t\!-u},\!y_{t\!-u}\!\left|y^{t\!-u-\!1}\right.\!\!\right)}\hspace{0.5cm}\nonumber\\
   &\!\stackrel{\textrm{(a)}}{=}&\!\alpha_{t\!-1}\!\!\left(x_{t\!-v}^{t\!-1},s_{t\!-v-\!1}^{t\!-u-\!1}\right)
     {\rm Pr}\!\left(x_t\!\left|x_{t\!-v}^{t\!-1},\!s_{t\!-v-\!1}^{t\!-u-\!1},\!y^{t\!-u-\!1}\right.\!\right) \nonumber\\
   && \! \times \, {\rm Pr}\!\left(y_{t\!-u}\!\left|x_{t\!-u},s_{t\!-u-\!1}\right.\!\right)
   \!{\rm Pr}\!\left(s_{t\!-u}\!\left|s_{t\!-u-\!1}\right.\!\right).
\end{eqnarray}
The equality (a) results from Proposition \ref{Propnewsource} and
the assumption $u \leq v$. From~(\ref{bcjr1}), we know that, once
the prior conditional probability vector
$\underline{\alpha}_{t-1}$ is given, the current conditional
probability vector $\underline{\alpha}_{t}$ depends {\em
only} on the current transition probability
${\rm Pr}\!\left(\!x_t\!\left|x_{t-v}^{t-1},\!s_{t-v-1}^{t-u-1},\!y^{t-u-1}\right.\!\right)$
and the channel transition law. To shorten the notation, we
abbreviate~(\ref{bcjr0}) and~(\ref{bcjr1}) as
\begin{equation}\label{bcjr}
 \underline{\alpha}_{t}\!=\!
 F_{\rm BCJR}\!\left(\underline{\alpha}_{t-\!1},
 \!\left\{{\rm Pr}\!\left(\!x_t\!\!\left|x_{t-\!v}^{t-\!1},\!s_{t-v-\!1}^{t-u-\!1},\!y^{t-u-\!1}\right.\!\right)\!\right\}\!,y_{t-u}\!\right)\!.
\end{equation}

Evidently, the vector $\underline{\alpha}_{t-1}$ depends on the sequence $y^{t-u-1}$, and two different sequences $y^{t-u-1}$ and $\tilde{y}^{t-u-1}$
may result in the same vectors $\underline{\alpha}_{t-1}$. For an arbitrarily selected source from $\mathcal{P}_v(u,u)$, two different sequences $y^{t-u-1}$ and $\tilde{y}^{t-u-1}$ may induce different probabilities
\begin{equation*}
{\rm Pr}\!\left(x_t\!\left|x_{t-v}^{t-1},s_{t-v-1}^{t-u-1},y^{t-u-1}\right.\!\right)\! \neq\!
{\rm Pr}\!\left(x_t\!\left|x_{t-v}^{t-1},s_{t-v-1}^{t-u-1},\tilde{y}^{t-u-1}\right.\!\right).
\end{equation*}
However, there do exist sources such that different sequences $y^{t-u-1}$ and $\tilde{y}^{t-u-1}$ resulting in the same vectors $\underline{\alpha}_{t-1}=\underline{\tilde{\alpha}}_{t-1}$ induce the same probabilities
\begin{equation*}
{\rm Pr}\!\left(x_t\!\left|x_{t-v}^{t-1},s_{t-v-1}^{t-u-1},y^{t-u-1}\right.\!\right)\! = \!
{\rm Pr}\!\left(x_t\!\left|x_{t-v}^{t-1},s_{t-v-1}^{t-u-1},\tilde{y}^{t-u-1}\right.\!\right).
\end{equation*}
Such a subclass of $\mathcal{P}_v(u,u)$ is defined as follows.

\begin{definition}\label{defCondMSalpha}
The set $\mathcal{P}'_v(u,u)$ collects all the $v$-th order
conditional Markov sources with $u$-delayed FB and $u$-delayed SI
such that
  \begin{equation*}
    {\rm Pr}\!\left(x_t\!\left|x_{t-v}^{t-1},s_{t-v-1}^{t-u-1},y^{t-u-1}\right.\!\right)\!=\!
    {\rm Pr}\!\left(x_t\!\left|x_{t-v}^{t-1},s_{t-v-1}^{t-u-1},\tilde{y}^{t-u-1}\right.\!\right)
  \end{equation*}
whenever ${\underline \alpha}_{t-1} = \tilde{\underline
 \alpha}_{t-1}$. Hence, the source set $\mathcal{P}'_v(u,u)$ can be
shortly denoted by
 \begin{equation*}
 \mathcal{P}'_v(u,u)\stackrel{\Delta}{=}
  \left\{{\rm Pr}\!\left(x_t\!\left|x_{t-v}^{t-1},s_{t-v-1}^{t-u-1},\underline{\alpha}_{t-1}\right.\!\right)\right\}_{t=1}^{\infty}.
 \end{equation*}
\vskip -0.25cm
\hfill \ding{113}
\end{definition}

Fig.~\ref{figFBFSCalpha} depicts the noncontrollable FSC model, whose
source belongs to the set $\mathcal{P}'_v(u,u)$.

\begin{figure}[!t]
\centering
\includegraphics[width=2.5in]{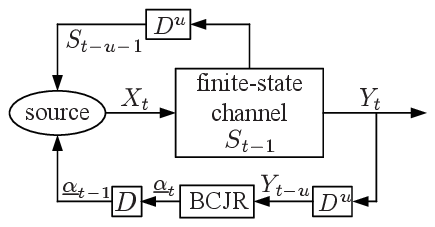}
\caption{A noncontrollable FSC whose source is in the set
$\mathcal{P}'_v(u,u)$.} \label{figFBFSCalpha}
\end{figure}

\begin{theorem}\label{Themnewsource2}
Let $u\!\leq\!v$. The supremum
$\mathcal{I}^{*}_{FB,SI}\!\left(u,\!v\right)$ can be achieved by a
source in the set $\mathcal{P}'_v(u,\!u)$, that is,
  \begin{equation}\label{rate3}
    \mathcal{I}^{*}_{FB,SI}(u,v)=\sup_{\mathcal{P}'_v(u,u)}
     \mathcal{I}_v\left(X,S\rightarrow Y\right)
  \end{equation}
where $\mathcal{P}'_v(u,u)=
  \left\{{\rm Pr}\!\left(x_t\!\left|x_{t-v}^{t-1},s_{t-v-1}^{t-u-1},\underline{\alpha}_{t-1}\right.\!\right)\right\}_{t=1}^{\infty}$.
\end{theorem}
\begin{IEEEproof}
See Appendix~\ref{proofThemnewsource2}.
\end{IEEEproof}

\section{DYNAMIC PROGRAMMING FOR SOURCE OPTIMIZATION} \label{sec5}

{ \subsection{Stochastic Control Formulations}
  \label{sec5.1}
}

From Theorem~\ref{Themnewsource2}, we only need to consider the sources in the set $\mathcal{P}'_v(u,u)$.
In this setting, for any given $y^{t-u-1}$,
\begin{eqnarray}\label{eqnIleftprob}
    \lefteqn{{\rm Pr}\!\left(x_{t-v}^t,s_{t-v-1},y_{t-u}^t\!\left|y^{t-u-1}\right.\right)} \nonumber\\
    &=& \sum_{s_{t-v}^{t-u}} {\rm Pr}\!\left(x_{t-v}^t,s_{t-v-1}^{t-u},y_{t-u}^t\!\left|y^{t-u-1}\right.\right)\nonumber\\
    &\stackrel{\rm (a)}{=}& \sum_{s_{t-v}^{t-u}}
        {\rm Pr}\!\left(\!x_{t\!-\!v}^{t\!-\!1},\!s_{t\!-\!v\!-\!1}^{t\!-\!u\!-\!1}\left|y^{t\!-\!u\!-\!1}\right.\!\right)
        {\rm Pr}\!\left(\!x_t\!\left|x_{t\!-\!v}^{t\!-\!1},\!s_{t\!-\!v\!-\!1}^{t\!-\!u\!-\!1},\!y^{t\!-\!u\!-\!1}\right.\!\right)\nonumber\\
    & & \:\:\:\:\:\:\:\times\;{\rm Pr}\!\left(y_{t\!-\!u}\left|x_{t\!-\!u},\!s_{t\!-\!u\!-\!1}\right.\!\right)
        {\rm Pr}\left(s_{t\!-\!u}\!\left|s_{t\!-\!u\!-\!1}\right.\!\right) \nonumber\\
    & & \:\:\:\:\:\:\:\times\;{\rm Pr}\!\left(y_{t\!-\!u+1}^t\!\left|x_{t\!-\!v}^t,\!s_{t\!-\!v\!-\!1}^{t\!-\!u}\right.\!\right)\nonumber\\
    &\stackrel{\rm (b)}{=}& \sum_{s_{t-v}^{t-u}}
        \alpha_{t\!-\!1}\!\!\left(\!x_{t\!-\!v}^{t\!-\!1},\!s_{t\!-\!v\!-\!1}^{t\!-\!u\!-\!1}\!\right)
        {\rm Pr}\!\left(\!x_t\!\left|x_{t\!-\!v}^{t\!-\!1},\!s_{t\!-\!v\!-\!1}^{t\!-\!u\!-\!1},\!\underline{\alpha}_{t-1}\right.\!\right)\nonumber\\
    & & \:\:\:\:\:\:\:\times\;{\rm Pr}\!\left(y_{t\!-\!u}\left|x_{t\!-\!u},\!s_{t\!-\!u\!-\!1}\right.\!\right)
        {\rm Pr}\left(s_{t\!-\!u}\!\left|s_{t\!-\!u\!-\!1}\right.\!\right)\nonumber\\
    & & \:\:\:\:\:\:\:\times\;{\rm Pr}\!\left(y_{t\!-\!u+1}^t\!\left|x_{t\!-\!v}^t,\!s_{t\!-\!v\!-\!1}^{t\!-\!u}\right.\!\right)
\end{eqnarray}
where equality (a) results from Proposition~\ref{Propnewsource} and
the assumption $u\leq v$, and equality (b) results directly from the definition of the source
set $\mathcal{P}'_v(u,u)$. Similar to equation~(\ref{eqnY}) as shown in Appendix B,
we can prove that the conditional probability ${\rm Pr}\!\left(y_{t-u+1}^t\left|x_{t-v}^t,s_{t-v-1}^{t-u}\right.\!\right)$
is completely determined by the channel law. Therefore, equalities in~(\ref{eqnIleftprob}) indicate
that the joint conditional probability mass function on the
left-hand side of~(\ref{eqnIleftprob}) is not sensitive to the
vector $y^{t-u-1}$~(that appears in the conditioning clause) but
to its induced variable $\underline{\alpha}_{t-1}$.
This implies that
\begin{eqnarray}\label{eqnIYytoalpha}
    \lefteqn{I\!\left(X_{t-v}^t,S_{t-v-1};Y_t\!\left|Y_{t-u+1}^{t-1},y_{t-u},y^{t-u-1}\right.\right)}\hspace{1.0cm} \nonumber\\
    &= I\!\left(X_{t-v}^t,S_{t-v-1};Y_t\!\left|Y_{t-u+1}^{t-1},y_{t-u},\underline{\alpha}_{t-1}\right.\right)
\end{eqnarray}
of which the right-hand side is a function of $\underline{\alpha}_{t-1}$,
$\left\{{\rm Pr}\!\left(x_{t}\!\left|x_{t-v}^{t-1},s_{t-v-1}^{t-u-1},\underline{\alpha}_{t-1}\right.\!\right)\right\}$ and $y_{t-u}$. For simplicity, we introduce the following notations
\begin{equation*}
    \begin{array}{ccc}
      {\rm p}_t(\underline{\alpha}_{t-1}) & \stackrel{\Delta}{=} &
        \left\{{\rm Pr}\!\left(x_{t}\!\left|x_{t-v}^{t-1},s_{t-v-1}^{t-u-1},\underline{\alpha}_{t-1}\right.\!\right)\right\} \\
      {\rm p}_t & \stackrel{\Delta}{=} & \left\{{\rm p}_t(\underline{\alpha}_{t-1}): \underline{\alpha}_{t-1} \in \mathcal{A}\right\}
    \end{array}.
\end{equation*}
Obviously, for $\underline{\alpha}_{t-1} \in \mathcal{A}$, the quantity ${\rm p}_t(\underline{\alpha}_{t-1})$ is a transition probability matrix of size $M \times |\mathcal{X}|$. Let $\mathscr{P}$ be the collection of all possible transition probability matrices. Both of the sets
$\mathcal{A}$ and $\mathscr{P}$ are bounded and closed, and hence compact.
Moreover, $\left\{\{{\rm p}_t\}_{t=1}^{\infty}\right\} = \{({\rm p}_1, {\rm p}_2, \cdots)\} = \mathcal{P}'_v(u,u)$.
Then the right-hand side of~(\ref{eqnIYytoalpha}) is a function that can be denoted by
\begin{eqnarray}\label{eqnReward}
    \lefteqn{g\left(\underline{\alpha}_{t-1}, {\rm p}_t(\underline{\alpha}_{t-1}), y_{t-u}\right)}\hspace{1.0cm} \nonumber\\
    &\stackrel{\Delta}{=} I\!\left(X_{t-v}^t,S_{t-v-1};Y_t\!\left|Y_{t-u+1}^{t-1},y_{t-u},\underline{\alpha}_{t-1}\right.\right).
\end{eqnarray}
Therefore, we can rewrite the directed information rate $\mathcal{I}_v(X,S\rightarrow Y)$ in (\ref{rate1}) as
\begin{eqnarray}\label{rate3too}
    \lefteqn{\mathcal{I}_v\!\left(X,S\rightarrow Y\right)} \nonumber\\
    &=&\!\!\!\liminf_{N\!\rightarrow\!\infty}\!\frac{1}{N}
       \sum_{t=1}^{N}I\!\left(X_{t-v}^t,S_{t-v-1};Y_t\!\left|Y_{t-u+1}^{t-1},Y^{t-u}\right.\!\right)\nonumber\\
    &=&\!\!\!\liminf_{N\!\rightarrow\!\infty}\!\frac{1}{N}\mathbf{E}\!
       \left[\sum_{t=1}^{N}g\left(\underline{\alpha}_{t-1},{\rm p}_t(\underline{\alpha}_{t-1}),Y_{t-u}\right)\!\right]\!.
\end{eqnarray}

Substituting~(\ref{rate3too}) into~(\ref{rate3}), we can see that the problem to find the upper bound
$\mathcal{I}^{*}_{FB,SI}(u,v)$ is equivalent to the following discrete-time
{\em infinite-horizon average reward per stage stochastic control
problem}~(ARSCP)~\cite{Bertsekas05,Bertsekas07,Arapostathis93},
which is referred to as {\bf Problem~A} for convenience.

{\bf Problem~A.}
The ARSCP is specified as follows.
\begin{enumerate}
    \item The {\em stochastic control system} of the problem is characterized by
        \begin{equation}\label{eqnSyst}
         \underline{\alpha}_{t}\!=\!
         F_{BCJR}\!\left(\underline{\alpha}_{t-1},{\rm p}_t(\underline{\alpha}_{t-1}),y_{t-u}\right)
        \end{equation}
        where
        \begin{enumerate}
            \item $\underline{\alpha}_{t-1}$ is the {\em state} and $\mathcal{A}$ is the {\em state space},
                i.e., $\underline{\alpha}_{t-1} \in \mathcal{A}$ and $\underline{\alpha}_{t} \in \mathcal{A}$;
            \item ${\rm p}_t$ is the function~(or {\em policy}) that maps the state space $\mathcal{A}$ to
                the {\em action} space $\mathscr{P}$, and ${\rm p}_t(\underline{\alpha}_{t-1})\in\mathscr{P}$ is the
                {\em policy}~(or {\em control}) when the state is $\underline{\alpha}_{t-1}$;
            \item $y_{t-u}$ is the \emph{disturbance}.
        \end{enumerate}

    \item The \emph{reward function} at stage $t$ is $g\!\left(\underline{\alpha}_{t-1}, {\rm p}_t(\underline{\alpha}_{t-1}), y_{t-u}\right)$. For convenience, we define the {\em expected reward function} at stage $t$ as
        \begin{eqnarray}\label{eqnEpReward}
            \lefteqn{g\left(\underline{\alpha}_{t-1}, {\rm p}_t(\underline{\alpha}_{t-1})\right)} \hspace{0.5cm}\nonumber\\
            &=& {\bf E}
            \left[g\left(\underline{\alpha}_{t-1}, {\rm p}_t(\underline{\alpha}_{t-1}), Y_{t-u}\right)\right] \nonumber\\
            &=& I\left(X_{t-v}^t,S_{t-v-1};Y_t\left|Y_{t-u}^{t-1},\underline{\alpha}_{t-1}\right.\right).
        \end{eqnarray}

    \item The objective of this problem is to find the {\em maximum average reward per stage}, i.e.,
        \begin{equation}\label{eqnProbA}
            \mathcal{I}(\underline{\alpha}_0\!)
            =\!\! \sup\limits_{\{{\rm p}_t\} \in \mathcal{P}'_v(u,u)} \!\!\!\!\!\!\mathcal{I}(\underline{\alpha}_0, \{{\rm p}_t\})
            ~{\rm for~all}~\underline{\alpha}_0 \!\!\in\! \mathcal{A}
        \end{equation}
        where $\mathcal{I}(\underline{\alpha}_0, \{{\rm p}_t\})$ is
        the average reward associated with the initial state $\underline{\alpha}_0$ and the sequence of
        policies $\{{\rm p}_t\}$
        \begin{equation}\label{eqnProbAaverRew}
            \mathcal{I}\!\left(\underline{\alpha}_0, \{{\rm p}_t\}\!\right)
            \!=\! \liminf_{N\rightarrow\infty}\!\frac{1}{N}\mathbf{E}\!
               \left[\sum\limits_{t=1}^{N}g\!\left(\underline{\alpha}_{t-1}, {\rm p}_t(\underline{\alpha}_{t-1}), Y_{t-u}\right)\!\right]\!\!.
        \end{equation}
\end{enumerate}

For the stochastic dynamic system~(\ref{eqnSyst}) of {\bf Problem A}, we have following two propositions.

\begin{proposition}\label{PropYtu}
The system disturbance variable $Y_{t-u}$ is characterized by a conditional probability distribution
that depends explicitly on the system
state $\underline{\alpha}_{t-1}$ and the policy
$\left\{{\rm Pr}\!\left(x_t\!\left|x_{t-v}^{t-1},s_{t-v-1}^{t-u-1},\underline{\alpha}_{t-1}\right.\!\right)\right\}$~\big(i.e., ${\rm p}_t(\underline{\alpha}_{t-1})$\big).
\end{proposition}

\begin{IEEEproof}
Given the system state $\underline{\alpha}_{t-1}$ and the policy $\left\{{\rm Pr}\!\left(x_t\!\left|x_{t-v}^{t-1},s_{t-v-1}^{t-u-1},\underline{\alpha}_{t-1}\right.\!\right)\right\}$,
the probability mass function of the system disturbance can be explicitly determined as
\begin{eqnarray}\label{eqnytu}
\lefteqn{{\rm Pr}\!\left(y_{t\!-\!u}\!\left|\underline{\alpha}_{t\!-\!1},\!
\left\{\!{\rm Pr}\!\left(x_t\!\left|x_{t\!-\!v}^{t\!-\!1},s_{t\!-\!v\!-\!1}^{t\!-\!u\!-\!1},
\underline{\alpha}_{t\!-\!1}\right.\!\right)\!\right\}\right.\!\right)}\nonumber\\
&=& \!\!\!\!\!\!\!\!\sum_{x_{t\!-\!v}^{t},s_{t\!-\!v\!-\!1}^{t\!-\!u}}\!\!\!\!\!
     {\rm Pr}\!\left(x_{t\!-\!v}^{t},s_{t\!-\!v\!-\!1}^{t\!-\!u},y_{t\!-\!u}\!\left|
     \underline{\alpha}_{t\!-\!1},\!\left\{\!{\rm Pr}\!\left(x_t\!\left|x_{t\!-\!v}^{t\!-\!1},s_{t\!-\!v-\!1}^{t\!-\!u\!-\!1},
     \underline{\alpha}_{t\!-\!1}\right.\!\right)\!\right\}\right.\!\right)\nonumber\\
 &\stackrel{\textrm{(a)}}{=}&\!\!\!\!\!\!\!\!\sum_{x_{t\!-\!v}^{t},s_{t\!-\!v\!-\!1}^{t\!-\!u}}\!\!\!\!\!
     \alpha_{t\!-\!1}\!\!\left(\!x_{t\!-\!v}^{t\!-\!1},\!s_{t\!-\!v\!-\!1}^{t\!-\!u\!-\!1}\!\right)
     {\rm Pr}\!\left(\!x_t\!\left|x_{t\!-\!v}^{t\!-\!1},\!s_{t\!-\!v\!-\!1}^{t\!-\!u\!-\!1},\!\underline{\alpha}_{t\!-\!1}\right.\!\right) \nonumber\\
  && \;\;\;\;\;\times\;{\rm Pr}\!\left(y_{t\!-\!u}\left|x_{t\!-\!u},\!s_{t\!-\!u\!-\!\!1}\right.\!\right)
    {\rm Pr}\!\left(s_{t\!-\!u}\!\left|s_{t\!-\!u\!-\!\!1}\right.\!\right)
\end{eqnarray}
where equality (a) follows from Proposition~\ref{Propnewsource} and the assumption $u\!\leq\!v$.
\end{IEEEproof}

\begin{proposition}\label{PropAMarkov}
The state process $\underline{A}_t$ with realization $\underline{\alpha}_t$ is a Markov process.
\end{proposition}

\begin{IEEEproof}
Equation~(\ref{eqnSyst}) and Proposition~\ref{PropYtu} imply that,
given the prior state $\underline{A}_{t-1}$, the current state
$\underline{A}_t$ is independent of the early states $\underline{A}_0^{t-2}$.
Hence, $\underline{A}_t$ is a Markov process.
\end{IEEEproof}

\begin{proposition}\label{PropApprox}
The reward function $g\!\left(\underline{\alpha}_{t\!-\!1}, {\rm p}_t(\underline{\alpha}_{t\!-\!1}\!), y_{t\!-\!u}\!\right)$
is uniformly continuous over $\mathcal{A} \times \mathscr{P}$.
\end{proposition}

\begin{IEEEproof}
This proposition can be proved by the compactness of the set $\mathcal{A} \times \mathscr{P}$ and the continuity of the reward function.
\end{IEEEproof}

In the average reward problem, i.e., {\bf Problem~A}, both the state $\underline{\alpha}$ and the policy ${\rm p}_t(\underline{\alpha})$ are continuous, which causes difficulties in theoretical analysis as well as computation. Fortunately, the uniform continuity of the reward function make it reasonable to restrict the reward function on discretized~(finite) state space and action space. This approach causes a loss at most $\varepsilon$ as long as the quantization is fine enough~\cite[Sec.~6.6]{Bertsekas05}\footnote{This holds for any continuous function $f(x)$ defined on a compact set $\Omega$.
Specifically, from the uniform continuity, for any $\varepsilon > 0$, there exists $\delta > 0$ such that $\|f(x_1) - f(x_2)\| \leq \varepsilon$ as long as $\|x_1 - x_2\| \leq \delta$, see~\cite{Browder96}. Now, we may take a quantizer $Q_{\delta}(\cdot)$ such that $\|x - Q_{\delta}(x)\| \leq \delta$.
Let $x^*$ and $\hat{x}$ be the solutions of the original problem $\max_{\Omega} f(x)$ and the discretized version $\max_{Q_{\delta}(\Omega)} f(x)$, respectively.
Then we have $f(\hat{x}) \geq f(Q_{\delta}(x^*)) \geq f(x^*) - \varepsilon$.}. That is, {\bf Problem~A} can be solved approximately~(resulting in an $\varepsilon$-optimal value) by solving its discretized version, {\bf Problem~B}.

{\bf Problem~B.} Let $\mathcal{Q}_{\delta}(\cdot)$ be a quantizer of the state set $\mathcal{A}$ which results in a finite-state space $\hat{\mathcal{A}} \subset \mathcal{A}$. Specifically, for any state $\underline{\alpha} \in \mathcal{A}$, there exists a quantized state $\hat{\underline{\alpha}} \in \hat{\mathcal{A}}$ such that the Euclidean distance satisfies $\|\underline{\alpha} - \hat{\underline{\alpha}} \| \leq \delta$ where $\delta$ is the designated quantization parameter. Similarly, let $\mathcal{Q}_{\xi}(\cdot)$ be the quantizer of the action space $\mathscr{P}$ and the resulting finite set be denoted by $\hat{\mathscr{P}}$.
The finite-state and finite-action ARSCP is specified as follows.
\begin{enumerate}
    \item The stochastic control system of this problem is
        \begin{equation}\label{eqnSystB}
            \hat{\underline{\alpha}}_t = \mathcal{Q}_{\delta} \left(F_{BCJR}(\hat{\underline{\alpha}}_{t-1}, \hat{\rm p}_t(\hat{\underline{\alpha}}_{t-1}), y_{t-u})\right)
        \end{equation}
        where
        \begin{enumerate}
            \item $\hat{\underline{\alpha}}_{t-1}$ is the state and $\hat{\mathcal{A}}$ is the state space;
            \item $\hat{{\rm p}}_t$ is the function~(or policy) that maps the state space $\hat{\mathcal{A}}$ to the action space $\hat{\mathscr{P}}$, and $\hat{{\rm p}}_t(\hat{\underline{\alpha}}_{t-1}) \in \hat{\mathscr{P}}$ is the policy when the state is $\hat{\underline{\alpha}}_{t-1}$;
            \item $y_{t-u}$ is the disturbance.
        \end{enumerate}

    \item The reward function at stage $t$ is $g\!\left(\hat{\underline{\alpha}}_{t-1}, \hat{\rm p}_t(\hat{\underline{\alpha}}_{t-1}), y_{t-u}\right)$.

    \item The objective of this problem is to find the maximum average reward per stage, i.e.,
        \begin{equation}\label{eqnProbB}
            \mathcal{I}(\hat{\underline{\alpha}}_0)
            \!=\!\! \sup_{\hat{\mathcal{P}}'_v(u,u)}\!\! \mathcal{I}\!\left(\hat{\underline{\alpha}}_0, \{\hat{\rm p}_t\}\right)
            ~{\rm for~all}~\hat{\underline{\alpha}}_0 \in \hat{\mathcal{A}}
        \end{equation}
        where
        \begin{itemize}
            \item $\hat{\mathcal{P}}'_v(u,u)$ is the collection of all policy sequences $\{\hat{\rm p}_t\}_{t=1}^{\infty}$ and is regarded as a discretized version of the source set $\mathcal{P}'_v(u,u)$;
            \item $\mathcal{I}(\hat{\underline{\alpha}}_0, \{\hat{\rm p}_t\})$ is the average reward associated with the initial state $\hat{\underline{\alpha}}_0$ and the sequence of policies $\{\hat{\rm p}_t\}$
                \begin{equation}\label{eqnProbBaverRew}
                    \!\!\!\!\!\!\!\mathcal{I}\!\left(\hat{\underline{\alpha}}_0,\!\{\hat{\rm p}_t\}\!\right)
                    \!=\! \liminf_{N\rightarrow\infty}\!\frac{1}{N}\mathbf{E}\!
                       \left[\sum\limits_{t=1}^{N}g\!\left(\hat{\underline{\alpha}}_{t-1}, \hat{\rm p}_t(\hat{\underline{\alpha}}_{t-1}\!),\!Y_{t-u}\right)\!\right]\!\!.
                \end{equation}
        \end{itemize}
\end{enumerate}

The {\em pair of coupled optimality equations}~\cite{Bertsekas05,Puterman94} of {\bf Problem~B} are
\begin{equation}\label{eqnOptEqn1}
    G^*(\underline{\alpha}) = \max_{{\rm p}(\underline{\alpha}) \in \hat{\mathscr{P}}}
    {\bf E}_{\underline{A}'|\underline{\alpha}}\left[G^*(\underline{A}')\right],
    ~{\rm for~any}~\underline{\alpha} \in \hat{\mathcal{A}}
\end{equation}
and
\begin{eqnarray}\label{eqnOptEqn2}
    \lefteqn{\!\!\!\!\!\!\!\!G^*\!(\underline{\alpha}) + J^*\!(\underline{\alpha})}\nonumber\\
    &\!\!\!\!\!\!\!\!\!\!\!\!\!\!\!\!\!\!\!\!\!\! = &\!\!\!\!\!\!\!\!\!\!\!\!\!\!\!\!
    \max\limits_{{\rm p}(\underline{\alpha}) \in \bar{\mathscr{P}}(\underline{\alpha})} \!\!\!
    \left\{g\!\left(\underline{\alpha}, {\rm p}(\underline{\alpha})\!\right)
    \!+\! \mathbf{E}_{\underline{A}'|\underline{\alpha}} \!\!\left[J^*\!(\underline{A}')\!\right]\!\right\},\!
    ~{\rm for~any}~\underline{\alpha} \!\in\! \hat{\mathcal{A}}
\end{eqnarray}
where $\bar{\mathscr{P}}(\underline{\alpha}) = \left\{{\rm p}(\underline{\alpha}): {\rm p}(\underline{\alpha}) \in \arg\max_{\hat{\mathscr{P}}} {\bf E}_{\underline{A}'|\underline{\alpha}}\left[{G}^*(\underline{A}')\right] \right\}$
is the set of policies attaining the maximum in equation~(\ref{eqnOptEqn1}).
The pair of coupled optimality equations can also be represented by vectors as
\begin{equation}\label{eqnOptEqnVec1}
    G^* = \max_{{\rm p} \in \mathcal{D}} L_{{\rm p}} G^*
\end{equation}
and
\begin{equation}\label{eqnOptEqnVec2}
    G^* + J^* = \max_{{\rm p} \in \bar{\mathcal{D}}}
    \left\{g + L_{{\rm p}} J^* \right\}
\end{equation}
where $\mathcal{D}$ is the set of all possible policies, i.e., $\mathcal{D} = \left\{{\rm p} = \left\{{\rm p}(\underline{\alpha}): \underline{\alpha} \in \hat{\mathcal{A}}\right\}\right\}$, and $\bar{\mathcal{D}}$ is the set of policies attaining the maximum in~(\ref{eqnOptEqnVec1}), i.e., $\bar{\mathcal{D}} = \left\{{\rm p} \in \mathcal{D}: {\rm p} \in \arg\max L_{{\rm p}} G^* \right\}$, and $L_{{\rm p}} = \left[{\rm Pr}(\underline{\alpha}'|\underline{\alpha}, {\rm p}(\underline{\alpha})) \right]_{|\hat{\mathcal{A}}| \times |\hat{\mathcal{A}}|}$ is a transition matrix between states under the policy ${\rm p}$. The solution $(G^*,J^*)$ to the pair of coupled optimality equations is usually called the {\em gain-bias pair}~\cite{Bertsekas07,Puterman94} with $G^*$ being the optimal average reward vector. The policy that achieves the maxima in the pair of coupled optimality equations is called the optimal policy.

{\bf Remark:} Depending on the choice of stationary policy, the Markov chain $\{\underline{A}_t \in \hat{\mathcal{A}}\}$ of {\bf Problem~B} may have different recurrent classes. Hence, {\bf Problem~B} is in general a {\em multi-chain} model~\cite{Bertsekas07}. The pair of coupled optimality equations of {\bf Problem~B} can be viewed as an analog to the Bellman equation for the {\em uni-chain} model~\cite{Bertsekas05,Bertsekas07}.

\begin{theorem}\label{ThemDP}
For {\bf Problem~B}, there exists a stationary policy that satisfies the pair of coupled optimality equations~(\ref{eqnOptEqn1}) and~(\ref{eqnOptEqn2}).
\end{theorem}
\begin{IEEEproof}
See Appendix~\ref{proofThemDP}.
\end{IEEEproof}

From Theorem~\ref{ThemDP}, it suffices to investigate only
stationary policies. For convenience, we denote
\begin{equation*}
    {\rm Pr}(j|i,\underline{\alpha})
    \!\stackrel{\Delta}{=}
    \!{\rm Pr}\!\left(X_t=j\!\left|\left(X_{t-v}^{t-1},S_{t-v-1}^{t-u-1}\!\right)=i, \underline{A}_{t-1}=\underline{\alpha}\right.\!\right).
\end{equation*}
Then the stationary policy in the discretized version of the source set $\mathcal{P}'_v(u,u)$ can be denoted by
\begin{equation*}
    {\rm p} = \left\{{\rm p}(\underline{\alpha}) = \left\{{\rm Pr}(j|i,\underline{\alpha})\right\} : \underline{\alpha} \in \hat{\mathcal{A}} \right\}.
\end{equation*}
We note that with a stationary source ${\rm p}$, the directed information rate $\mathcal{I}_v(X,S \rightarrow Y)$ in~(\ref{rate3too}) can be computed using Monte Carlo methods
similar to those in~\cite{Arnold01,Pfister01,Sharma01,Arnold06}.

{\subsection{A Value Iteration Method to Solve {\bf Problem~B}}\label{sec5.2} }

For a finite-state and finite-action ARSCP, there exist several
dynamic programming algorithms~(such as value iteration, policy iteration and
linear programming)~\cite{Bertsekas07} to solve the pair of coupled optimality equations. To obtain $\varepsilon$-optimal value with small $\varepsilon$, fine quantization is required, but then the discretized state space $\hat{\mathcal{A}}$ and action space $\hat{\mathscr{P}}$ usually have large sizes. In this setting, the value iteration method is a better choice. In this subsection, a value iteration algorithm is introduced to solve {\bf Problem~B}. Under a mild assumption, the presented value iteration algorithm is shown to be convergent and delivers the near-optimal stationary policy and the optimal average reward value numerically.

The value iteration method is, for all $\underline{\alpha} \in \hat{\mathcal{A}}$,
\begin{equation}\label{eqnValIter1}
    J_k(\underline{\alpha})
    = \max_{\hat{\mathscr{P}}}
        \left\{g(\underline{\alpha}, {\rm p}(\underline{\alpha}))
        + {\bf E}_{\underline{A}'|\underline{\alpha}} \left[J_{k-1}(\underline{A}')\right] \right\}
\end{equation}
starting from an arbitrary initial function $J_0$.
In the following, we show that this value iteration method can deliver a solution $(G^*, J^*)$ to the pair of coupled optimality equations~(\ref{eqnOptEqn1}) and~(\ref{eqnOptEqn2}). On one hand, from Proposition~4.3.1 in~\cite{Bertsekas07}, the optimal average reward vector $G^*$ can be obtained as
\begin{equation}\label{eqnGoptAR1}
    G^* = \lim_{k\rightarrow \infty} \frac{J_k}{k}.
\end{equation}
Note that in general, for a multi-chain average reward problem, $G^*(\underline{\alpha})$ may be different for different $\underline{\alpha}$. But by performing the iteration method for Example~\ref{GE}, we find that the values $\frac{J_k(\underline{\alpha})}{k}$ are always numerically approaching a constant as $k \rightarrow \infty$.

On the other hand, we need to find $J^*$. To this end, we make an additional assumption as follows.

\begin{assumption}\label{AssumpAperdic}
Every optimal stationary policy $\rm p$ has an aperiodic transition probability matrix $L_{\rm p}$.
\hfill \ding{113}
\end{assumption}

{\bf Remark:} Recall that
\begin{equation}
    \underline{\alpha}_{t} = \left[\alpha_{t}(0),\alpha_{t}(1),\cdots,\alpha_{t}(M-1)\right]
 \end{equation}
and
 \begin{equation}
  \alpha_{t}(\ell) = {\rm Pr}\!\left(\left(\!X_{t-v+1}^{t},\!S_{t-v}^{t-u}\!\right)\!=\!\ell\,|y^{t-u}\!\right).
 \end{equation}
Intuitively, the {\em optimal} stationary policy should not depend heavily on the early channel outputs. In other words, the influence of $y^{t-w-1}$ on the optimal policy should die away with sufficiently large $w$. Specifically, for two different channel output sequences $(y^{t-w-1}, y_{t-w}^{t-u})$ and $(\tilde{y}^{t-w-1}, y_{t-w}^{t-u})$, the resulting probability vectors $\underline{\alpha}_t$ and ${\tilde{\underline{\alpha}}_t}$ should be almost the same~(i.e., their Euclidean distance should be very small). As a result, the quantized versions of $\underline{\alpha}_t$ and ${\underline{\tilde{\alpha}}}_t$ will be equal. This implies that, for a given optimal stationary policy, the states $\underline{A}_t \in \mathcal{A}$ can be restricted to the subset of states~(called it the subset of {\em effective} states) that correspond to the most recent channel outputs $Y_{t-w}^{t-u}$. Such a subset is {\em communicative}. In particular, the state $\underline{\alpha}$ corresponding to the vector $Y_{t-w}^{t-u} = \underline{0}$ can be reached from itself whenever the next channel output $Y_{t-u+1}$ equals $0$. Hence, the Markov chain is essentially aperiodic. This intuition has also been verified numerically in our example.

Under Assumption~\ref{AssumpAperdic}, according to Propositions~4.3.5 and~4.3.6 in~\cite{Bertsekas07}, we have the following facts.
\begin{enumerate}
    \item The optimal average reward vector $G^*$ satisfying~(\ref{eqnGoptAR1}) can also be obtained by
        \begin{equation}
            G^* = \lim_{k\rightarrow \infty} (J_{k} - J_{k-1}).
        \end{equation}
    \item The bias $J^*$ can be obtained by
        \begin{equation}
            J^* = \lim_{k\rightarrow \infty} (J_k - kG^*).
        \end{equation}
    \item There exists a sufficiently large $K$ such that for any $k \geq K$,
        \begin{eqnarray}\label{eqnTwoMax}
            \lefteqn{\!\!\!\!\!\!\!\!\max_{{\rm p}(\underline{\alpha}) \in \hat{\mathscr{P}}}
            \!\left\{g(\underline{\alpha}, {\rm p}(\underline{\alpha}))
            + {\bf E}_{\underline{A}'|\underline{\alpha}} \left[J_{k-1}(\underline{A}')\right] \right\}} \nonumber\\
            &\!\!\!\!=&\!\!\!\!\!\!\!\! \max\limits_{{\rm p}(\underline{\alpha}) \in \bar{\mathscr{P}}(\underline{\alpha})}
                \!\left\{g(\underline{\alpha}, {\rm p}(\underline{\alpha}))
                + {\bf E}_{\underline{A}'|\underline{\alpha}} \left[J_{k-1}(\underline{A}')\right] \right\}
        \end{eqnarray}
        where $\bar{\mathscr{P}}(\underline{\alpha})$ has been defined in the previous subsection, see equation~(\ref{eqnOptEqn2}).
\end{enumerate}
Therefore, the pair $(G^*,J^*)$ induced by the value iteration method~(\ref{eqnValIter1}) is a solution to the pair of coupled optimality equations~(\ref{eqnOptEqn1}) and~(\ref{eqnOptEqn2}). Moreover, let ${\rm p}$ be the policy obtained by the value iteration method~(\ref{eqnValIter1}) for the sufficiently large $K$. Then $\{{\rm p}\}^{\infty}$ can achieve numerically optimal average reward value of {\bf Problem~B}. A practical value iteration algorithm for {\bf Problem~B} is described as follows.

\begin{algorithm}[A Value Iteration Algorithm]\label{AlgProbB}
\hspace{\parindent}
\begin{enumerate}
    \item \textbf{Initialization:}
        \begin{itemize}
            \item Choose a large positive integer $n$.
            \item Initialize the {\em terminal reward function} or {\em starting vector} as
            $J_0(\underline{\alpha}) = 0$ for all $\underline{\alpha} \in \hat{\mathcal{A}}$.
        \end{itemize}
    \item \textbf{Recursions:}\\
        For $k\!=\!1,2,\ldots,n$, and any $\underline{\alpha} \in \hat{\mathcal{A}}$, compute
        \begin{equation}\label{VIA}
            J_{k}(\underline{\alpha})\!=\!\!\!
            \max_{{\rm p}(\underline{\alpha}) \in \hat{\mathscr{P}}}\!\!
            \left\{g\!\left(\underline{\alpha},\!{\rm p}(\underline{\alpha})\right)\!
                +\!\mathbf{E}_{\underline{A}'|\underline{\alpha}} \left[J_{k-1}(\underline{A}')\right]\!\right\}\!.
        \end{equation}
        where $\underline{A}' \in \hat{\mathcal{A}}$ is the random variable that depends on the system disturbance variable $Y_{t-u}$,
        and where the realization $\underline{\alpha}'$ of $\underline{A}'$ can be computed by
        \begin{equation}\label{bcjrqua}
            \underline{\alpha}'\!= \mathcal{Q}_{\delta}\!\left(F_{BCJR}\!\left(\underline{\alpha},{\rm p}(\underline{\alpha}),
            y_{t-u}\right)\right).
        \end{equation}
    \item \textbf{Optimized source:}\\
        For any $\underline{\alpha} \in \hat{\mathcal{A}}$,
        the optimized source distribution is delivered as
            \begin{equation}\label{VIApolocy}
                {\rm p}^*(\underline{\alpha}) =
                \arg\!\max_{{\rm p}(\underline{\alpha}) \in \hat{\mathscr{P}}}\!
                \left\{\!g\!\left(\underline{\alpha},\!{\rm p}(\underline{\alpha})\right)\!
                +\!\!\mathbf{E}_{\underline{A}'|\underline{\alpha}} \left[J_{n}(\underline{A}')\right]\!\right\}.
            \end{equation}
    \item \textbf{End.}
\end{enumerate}
\end{algorithm}

{\bf Remark:} By implementing Algorithm~\ref{AlgProbB}, we can obtain stationary Markov source probabilities ${\rm p}^* = \left\{{\rm p}^*(\underline{\alpha}): \underline{\alpha} \in \hat{\mathcal{A}}\right\}$, which can be utilized to evaluate numerically the optimal average reward of {\bf Problem~B}, i.e., the $\varepsilon$-optimal value of {\bf Problem~A}. Strictly speaking, the optimal stationary policy ${\rm p}^*$ obtained in~(\ref{VIApolocy}) for {\bf Problem~B} is an approximation of the optimal stationary policy of {\bf Problem~A}, and the information rate $\mathcal{I}_v\left(X,S\rightarrow Y\right)$ induced by the ``optimal'' stationary policy
${\rm p}^*$ is only a lower bound on the upper bound $\mathcal{I}_{FB,SI}^*(u,v)$. Obviously, finer quantization of $\mathcal{A}$ and $\mathscr{P}$ should cause less loss of optimality. The numerical values resulting from different quantizations are discussed in the following section.

\section{Numerical results}\label{sec6}

\begin{figure}[!t]
\centering 
\includegraphics[width=3.3in]{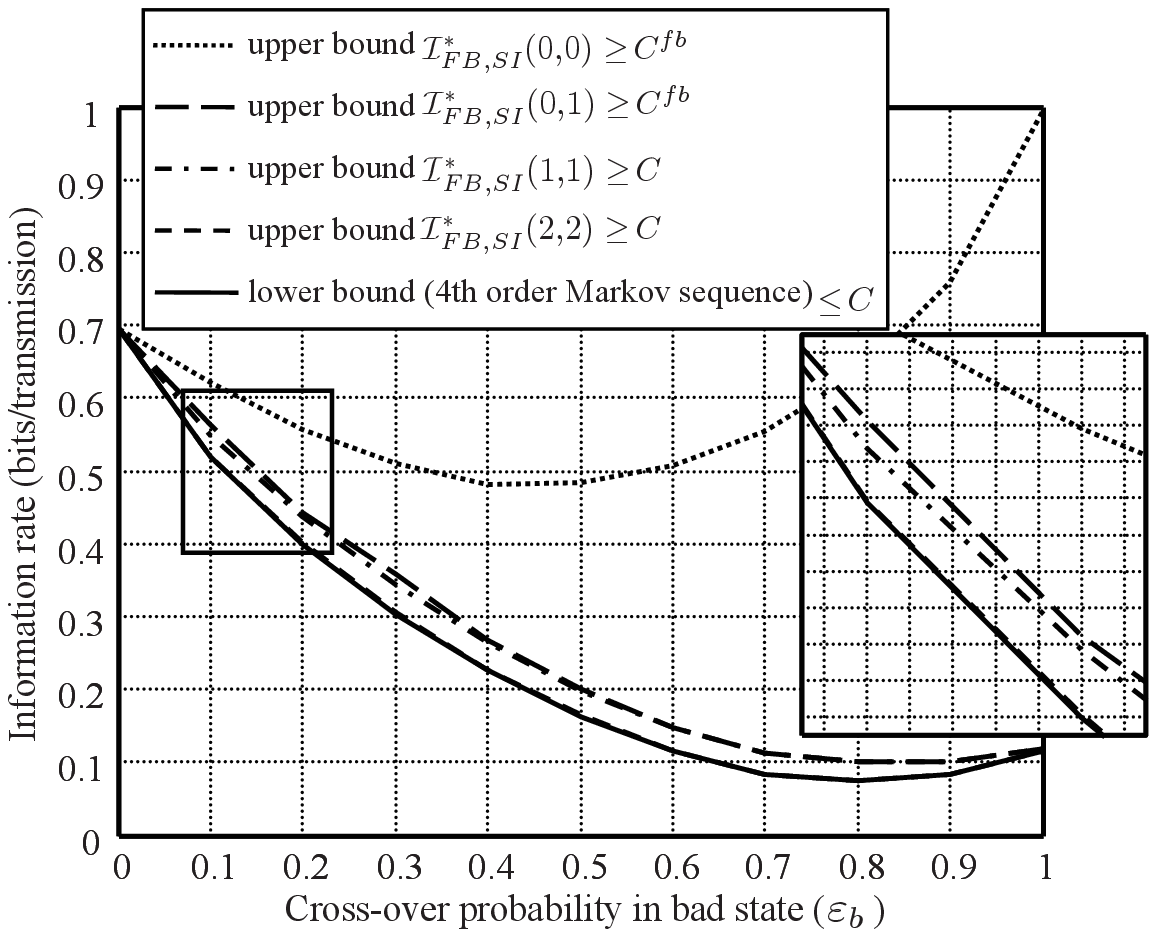}
\caption{Bounds on the capacities of the RLL$(1,\infty)$-GE
channel.} \label{bound}
\end{figure}

In this section, we present numerical results by taking
the RLL$(1, \infty)$-GE channel\footnote{Note that restricting
the input as RLL$(1, \infty)$ sequence is equivalent to restricting
certain transition probabilities to be zeros. Since the action space is
still compact, the results in Sections IV and V can be applied here.}
shown in Fig.~\ref{RLL} and Fig.~\ref{GEg} as an example.
We chose this channel because it was already used in a prior
publication~\cite{Vontobel08}.
In this example, we set the
transition probabilities between the channel states as
$p(b\!\left|g\right.)\!=\!p(g\!\left|b\right.)\!=\!0.3$, the
cross-over probability in the ``good'' state as
$\varepsilon_g\!=\!0.001$ and the cross-over probability in the ``bad''
state as a variable $\varepsilon_b\in [0, 1]$.
Firstly, we quantize the state space $\mathcal{A}$
and the action space $\mathscr{P}$ using parameters $\delta$ and $\xi$, respectively. Secondly, we apply
Algorithm~\ref{AlgProbB} introduced in Section \ref{sec5} to
obtain an ``optimal'' stationary policy. Finally, we use
Monte Carlo methods~\cite{Arnold01,Pfister01,Sharma01,Arnold06}
to numerically evaluate the upper bounds $\mathcal{I}^*_{FB,SI}(u,v)$.
The results are shown in Fig.~\ref{bound}, where
$\mathcal{I}^*_{FB,SI}(1,1)$ and $\mathcal{I}^*_{FB,SI}(2,2)$ are
two upper bounds on the feedforward capacity, and
$\mathcal{I}^*_{FB,SI}(0,0)$ and $\mathcal{I}^*_{FB,SI}(0,1)$ are
two upper bounds on the feedback capacity. As expected,
$\mathcal{I}^*_{FB,SI}(2,2) \leq \mathcal{I}^*_{FB,SI}(1,1) \leq
\mathcal{I}^*_{FB,SI}(0,1) \leq \mathcal{I}^*_{FB,SI}(0,0)$. It is
worth pointing out that, due to the RLL constraints, the source
must have memory of order at least one and the optimization is
implemented by taking into account the RLL constraint. In
particular, the upper bound $\mathcal{I}^*_{FB,SI}(0,0)$ is
obtained by optimizing the sources $\mathcal{P}'_1(0, 0)$.
Also shown in Fig.~\ref{bound} is a lower bound on $C$ computed using
techniques presented in~\cite{Kavcic01,Vontobel08}. By comparing
$\mathcal{I}^*_{FB,SI}(2,2)$ with the lower bound, we observe that
the bounds $\mathcal{I}^*_{FB,SI}(v,v)$ are numerically tight
upper bounds on the feedforward capacity. We are unable to
evaluate the tightness of the upper bounds
$\mathcal{I}^*_{FB,SI}(0,v)$ on the feedback capacity since no
good lower bounds on $C^{fb}$ are available in the literature for
noncontrollable FSCs.

Fig.~\ref{FigQuantizer} illustrates the loss of the optimality
caused by quantization. We focus on the computation of
$\mathcal{I}^*_{FB,SI}(1,1)$. Let the quantization parameter of
the action space $\mathscr{P}$ be fixed, i.e., $\xi = 0.0125$, and
the quantization parameter $\delta$ of the state space $\mathcal{A}$ be varying.
From Fig.~\ref{FigQuantizer}, we can see that a smaller $\delta$
(equivalently, a finer quantizer) induces a larger information
rate $\mathcal{I}_1(X,S \rightarrow Y)$ and causes less loss of
optimality. It can also be seen that the gap between the different
quantizers is negligible for small quantization parameters
$\delta$.

\begin{figure}[!t]
\centering
\includegraphics[width=3.0in]{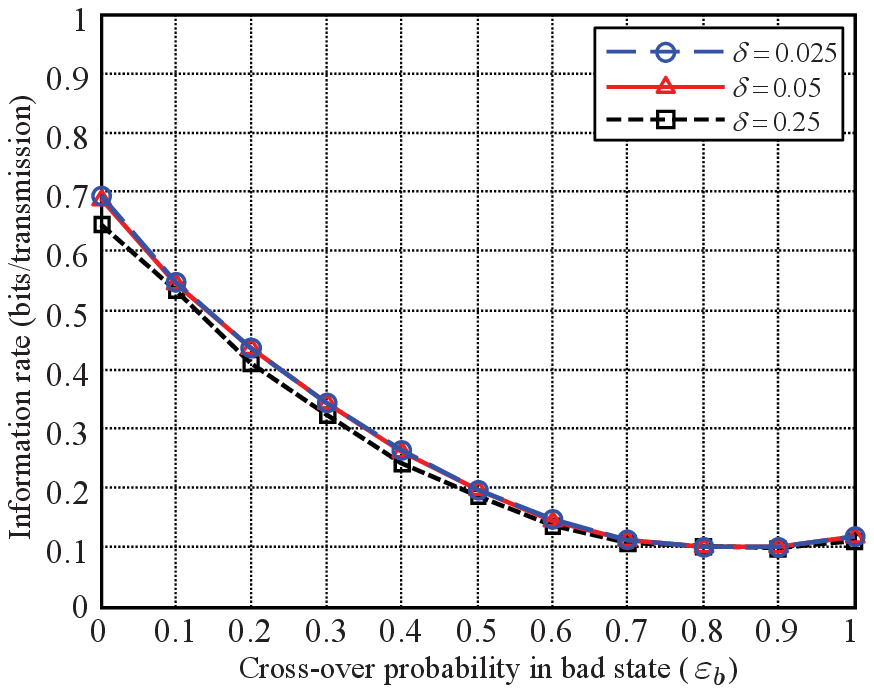}\\
    \caption{Information rates $\mathcal{I}_1(X,S\rightarrow Y)$ for ``optimal'' quantized sources in $\mathcal{P}'_1(1,1)$ delivered by Algorithm~\ref{AlgProbB} with different quantizers, where the quantization parameters of the state space and the action space are $\delta$ and $\xi = 0.0125$, respectively.}\label{FigQuantizer}
\end{figure}

\section{Conclusion}\label{sec7}


By the technique of inserting the delayed channel state into the channel input, the directed information rate from the new channel input~(including the channel input and the delayed channel state) to the channel output is defined, and then a universal form of upper bounds on the capacities of the noncontrollable FSC has been developed. In particular, two respective nested sequences of upper bounds on the feedforward capacity and the feedback capacity are obtained. It has been shown that these upper bounds can be achieved by finite order conditional Markov sources with delayed output feedback~(FB) and delayed state information~(SI). Moreover, the computation of the upper bounds was formulated as an average reward per stage stochastic control problem~(ARSCP) with a continuous state space and a continuous action space. By the compactness of the state space and the action space and the unform continuity of the reward function, the original ARSCP was transformed into an ARSCP with a finite state set and a finite action set, which can be solved by a value iteration algorithm. Under a mild assumption, the value iteration algorithm is shown to be convergent and delivers a near-optimal stationary policy as well as numerically tight upper bounds.

\appendices

\section{Proof of Theorem~\ref{ThemCap}}\label{proofThemCap}

\begin{IEEEproof}
The feedforward capacity in~(\ref{eqnCap}) and the feedback capacity in~(\ref{eqnCapFB}) are rewritten as
\begin{equation}\label{eqnCapRe}
  C = \sup_{\left\{{\rm Pr}\left(x_t\left|x^{t-1}\right.\!\right)\right\}_{t=1}^{\infty}}
  \liminf_{N \rightarrow \infty} \frac{1}{N} I(X^N \rightarrow Y^N)
\end{equation}
and
\begin{equation}\label{eqnCapFBRe}
  C^{fb} = \sup_{\left\{{\rm Pr}\left(x_t\left|x^{t-1},y^{t-1}\right.\!\right)\right\}_{t=1}^{\infty}}
  \liminf_{N \rightarrow \infty} \frac{1}{N} I(X^N \rightarrow Y^N),
\end{equation}
respectively.
We now prove that they are equal to the capacities
\begin{equation}\label{eqnCapGall}
    C_{G} = \lim_{N \rightarrow \infty} \sup_{\left\{{\rm Pr}(x_t|x^{t-1})\right\}_{t=1}^N} \frac{1}{N} I(X^N \rightarrow Y^N)
\end{equation}
defined by Gallager in~\cite[Theorems~4.6.4 and~5.9.1]{Gallager68}
and
\begin{equation}\label{eqnCapFBKimPert}
    C_P^{fb} = \lim_{N \rightarrow \infty} \sup_{\left\{{\rm Pr}(x_t|x^{t-1},y^{t-1})\right\}_{t=1}^N} \frac{1}{N} I(X^N \rightarrow Y^N)
\end{equation}
defined by Permuter~{\em et~al.} in~\cite[Theorem~18]{Permuter09}, respectively.
Here, we only prove $C = C_G$. A similar method~(omitted here) can be used to prove $C^{fb} = C_P^{fb}$.

On one hand, we have $C \leq C_G$.
Let $\left\{{\rm Pr}^*\!\left(x_t\left|x^{t-1}\right.\!\right)\right\}_{t=1}^{\infty}$ be a sequence of sources that achieves the capacity $C$. Then, for each $N$ and the fixed sequence $\left\{{\rm Pr}^*\!\left(x_t\left|x^{t-1}\right.\!\right)\right\}_{t=1}^{N}$, the corresponding directed information $I^*(X^N \rightarrow Y^N)$ is less than $\sup I(X^N \rightarrow Y^N)$, which implies that $C \leq C_G$.

On the other hand, we prove $C_G \leq C$. To this end, we introduce a new capacity expression
\vskip -0.6cm
\begin{equation}\label{eqnCapMid}
    C_M = \sup_{\left\{\left\{{\rm Pr}(x_t|x^{t-1})\right\}_{t=1}^T\right\}_{T=1}^{\infty}} \liminf_{N \rightarrow \infty} \frac{1}{N} I(X^N \rightarrow Y^N)
    \vspace{-0.1cm}
\end{equation}
where the supremum is taken over all possible sequences of sources without the consistency requirement, i.e., $\left\{\left\{{\rm Pr}(x_t|x^{t-1})\right\}_{t=1}^T\right\}_{T=1}^{\infty}$. Firstly, we prove that $C_G \leq C_M$. For each $N$, denote the optimal source achieving $\sup I(X^N\rightarrow Y^N)$ as $\left\{{\rm Pr}^* (x_t|x^{t-1})\right\}_{t=1}^N$. For the fixed sequence of sources $\left\{\left\{{\rm Pr}^* (x_t|x^{t-1})\right\}_{t=1}^T\right\}_{T=1}^{\infty}$, $\liminf \frac{1}{N} I(X^N\rightarrow Y^N) = C_G$ trivially holds. Thus we have $C_G \leq C_M$.
Secondly, we prove that $C_M = C$. It is obvious that $C \leq C_M$ since $\left\{{\rm Pr}(x_t|x^{t-1})\right\}_{t=1}^{\infty} \subset \left\{\left\{{\rm Pr}(x_t|x^{t-1})\right\}_{t=1}^T\right\}_{T=1}^{\infty}$.
Now we need to prove that $C < C_M$ does not hold. Otherwise, there must exist a sequence of sources $\left\{\left\{{\rm {\tilde Pr}} (x_t|x^{t-1})\right\}_{t=1}^T\right\}_{T=1}^{\infty}$ such that $\liminf_{N \rightarrow \infty} \frac{1}{N} I(X^N \rightarrow Y^N) \geq C + \epsilon_0$ where $\epsilon_0 > 0$. It implies that there exists a $K$ such that for all $N \geq K$, $\frac{1}{N} I(X^N \rightarrow Y^N) \geq C + \epsilon$ where $0 < \epsilon < \epsilon_0$. Let $\left\{{\rm {\tilde Pr}}(x_t|x^{t-1})\right\}_{t=1}^N$ be the source for a fixed $N \geq K$. Construct a process by $X^{\infty} = X^N \times X^N \times \cdots$ with probability assignment ${\rm Pr}(x^{\infty}) = \left({\rm Pr}(x^N)\right)^{\infty}$. Consider the directed information rate $\frac{1}{NL} I(X^{NL} \rightarrow Y^{NL})$.
\vskip -0.4cm
\begin{eqnarray}
    \lefteqn{\frac{1}{NL} I(X^{NL} \rightarrow Y^{NL})}\nonumber\\
    &\!\!\!\!\!=& \!\!\!\!\!\!\frac{1}{NL} \sum_{i = 1}^{NL}  I(X^{i}; Y_i|Y^{i-1}) \nonumber\\
    &\!\!\!\!\!=& \!\!\!\!\!\!\frac{1}{NL} \sum_{\ell = 0}^{L-1} \sum_{i=1}^N I(X^{\ell N +i}; Y_{\ell N +i}|Y^{\ell N +i -1}) \nonumber\\
    &\!\!\!\!\!\geq& \!\!\!\!\!\!\frac{1}{NL} \sum_{\ell = 0}^{L-1} \sum_{i=1}^N I(X_{\ell N +1}^{\ell N +i}; Y_{\ell N +i}|Y^{\ell N +i -1}) \nonumber\\
    &\!\!\!\!\!=& \!\!\!\!\!\!\frac{1}{NL} \sum_{\ell = 0}^{L-1} \sum_{i=1}^N I(X_{\ell N +1}^{\ell N +i}; Y_{\ell N +i}|Y_{\ell N +1}^{\ell N +i -1}, Y^{\ell N}) \nonumber\\
    &\!\!\!\!\!\stackrel{\rm (a)}{\geq}& \!\!\!\!\!\!\frac{1}{NL} \sum_{\ell = 0}^{L-1} \left(\!\!\!\begin{array}{l}
                                                                           -\log|\mathcal{S}| \\
                                                                           + \sum\limits_{i=1}^N I(X_{\ell N +1}^{\ell N +i}; Y_{\ell N +i}|Y_{\ell N +1}^{\ell N +i -1}, Y^{\ell N}, S_{\ell N})
                                                                         \end{array}
    \!\!\!\!\!\right)\nonumber\\
    &\!\!\!\!\!\stackrel{\rm (b)}{=}& \!\!\!\!\!\!\frac{1}{NL} \sum_{\ell = 0}^{L-1}
        \left(\!\!-\log|\mathcal{S}| + \sum_{i=1}^N I(X_{\ell N +1}^{\ell N +i}; Y_{\ell N +i}|Y_{\ell N +1}^{\ell N +i -1}\!, \!S_{\ell N})\!\!\right)\nonumber\\
    &\!\!\!\!\!\stackrel{\rm (c)}{\geq}& \!\!\!\!\!\!\frac{1}{NL} \sum_{\ell = 0}^{L-1}
        \left(\!\!-2\log|\mathcal{S}| + \sum_{i=1}^N I(X_{\ell N +1}^{\ell N +i}; Y_{\ell N +i}|Y_{\ell N +1}^{\ell N +i -1})\!\!\right)\nonumber\\
    &\!\!\!\!\!=& \!\!\!\!\!\!\frac{1}{NL} \sum_{\ell = 0}^{L-1} \left( -2\log|\mathcal{S}| + I(X_{\ell N +1}^{\ell N + N} \rightarrow Y_{\ell N +1}^{\ell N +N}) \right) \nonumber\\
    &\!\!\!\!\!\stackrel{\rm (d)}{=}& \!\!\!\!\!\!\frac{1}{N} \left( -2\log|\mathcal{S}| + I(X^N \rightarrow Y^N)\right) \nonumber\\
    &\!\!\!\!\!\geq& \!\!\!\!\!\!C + \epsilon - \frac{2}{N}\log|\mathcal{S}|
\end{eqnarray}
where inequalities (a) and (c) result from Lemma~4 in~[13],
equality (b) results from the Markovianity of the chain $(X^{\ell N}, Y^{\ell N}) \rightarrow S_{\ell N} \rightarrow (X_{\ell N +1}^{(\ell +1) N}, Y_{\ell N +1}^{(\ell +1) N})$,
and equality (d) results from the assumptions of channel model and the construction of the process which imply that
$I(X_{\ell N +1}^{\ell N + N} \rightarrow Y_{\ell N +1}^{\ell N +N}) = I(X^N \rightarrow Y^N)$ for all $\ell$. By the choice of $\epsilon_0$ and $\epsilon$, for any $L$, $\frac{1}{NL} I(X^{NL} \rightarrow Y^{NL}) > C + \delta$ where $\delta > 0$. Then $\liminf \frac{1}{N} I(X^{N} \rightarrow Y^{N}) > C$, which raises a contradiction, regarding the expression of $C$ in~(\ref{eqnCapRe}).
Therefore, $C_G \leq C_M = C$.
\end{IEEEproof}

\section{Proof of Theorem~\ref{Themnewsource1}}\label{proofThemnewsource1}

\begin{IEEEproof}
Let $\mathcal{P}_1\!\in\!\mathcal{P}(u,u)$ be an arbitrary source
with $u$-delayed FB and $u$-delayed SI. Denote the corresponding
information as $I\left(\left.X_{t-v}^t,S_{t-v-1};Y_t\right|Y^{t-1}\right)$.
To prove Theorem~\ref{Themnewsource1}, it is sufficient to show that there
exists a conditional Markov source $\mathcal{P}_2$ in
$\mathcal{P}_v(u,u)\subseteq \mathcal{P}(u,u)$ with the same information
$I\left(\left.X_{t-v}^t,S_{t-v-1};Y_t\right|Y^{t-1}\right)$
as that achieved by $\mathcal{P}_1$. To do this,
for any given $\mathcal{P}_1\!\in\!\mathcal{P}(u,u)$, we
construct a new source $\mathcal{P}_2\in \mathcal{P}_v(u,u)$ as
\begin{eqnarray}\label{eqnbuiltsource}
 \lefteqn{{\rm Pr}^{\left(\mathcal{P}_2\right)}\left(x_t\left|x^{t-1},s_0^{t-u-1},y^{t-u-1}\right.\right)} \hspace{1.8cm} \nonumber \\
 &\stackrel{\Delta}{=}
 {\rm Pr}^{\left(\mathcal{P}_1\right)}\left(x_t\left|x_{t-v}^{t-1},s_{t-v-1}^{t-u-1},y^{t-u-1}\right.\right)
\end{eqnarray}
with the initial probability as
\begin{equation*}
    {\rm Pr}^{(\mathcal{P}_2)}\!\left(x^{v},s_0^{v-u},y^{v-u}\!\right)\!\stackrel{\Delta}{=}\!
    {\rm Pr}^{(\mathcal{P}_1)}\!\left(x^{v},s_0^{v-u},y^{v-u}\!\right).
\end{equation*}

In the following, we will prove that both $\mathcal{P}_1$ and
$\mathcal{P}_2$ induce the same joint probability distribution
${\rm Pr}\!\left(x_{t-v}^t,\!s_{t-v-1},\!y^t\right)$, which,
together with the result of Theorem~\ref{ThemshortS}, completes the proof
of Theorem~\ref{Themnewsource1}.

Actually, for any source with $u$-delayed FB and $u$-delayed SI,
we have
\begin{eqnarray}\label{sourceG}
  \lefteqn{{\rm Pr}\!\left(x_{t-v}^t,\!s_{t-v-1},\!y^t\right)}\nonumber\\
  &\!\!\!\!\!\!=&\!\!\!\!\!\!\!\!\! \sum_{x^{t\!-\!v\!-\!1},s_0^{t\!-\!v\!-\!2},s_{t\!-\!v}^{t\!-\!u}}\!\!\!\!\!\!\!\!\!
     \!{\rm Pr}\!\left(x^t,s_0^{t-u},y^{t}\right)\nonumber\\
  &\!\!\!\!\!\!=&\!\!\!\!\!\!\!\!\! \sum_{x^{t\!-\!v\!-\!1},s_0^{t\!-\!v\!-\!2},s_{t\!-\!v}^{t\!-\!u}}\!\!\!\!\!\!\!\!\!
     \!{\rm Pr}\!\left(x^t,s_0^{t-u},y^{t-u}\right)
     {\rm Pr}\!\left(\left.y_{t-u+1}^t\right|x^t,s_0^{t-u},y^{t-u}\right)\nonumber\\
  &\!\!\!\!\!\!=&\!\!\!\!\!\!\!\!\! \sum_{x^{t\!-\!v\!-\!1},s_0^{t\!-\!v\!-\!2},s_{t\!-\!v}^{t\!-\!u}}\!\!\!\!\!\!\!\!\!
     \!{\rm Pr}\!\left(x^v,s_0^{v-u},y^{v-u}\right)
     {\rm Pr}\!\left(\left.y_{t-u+1}^t\right|x^t,s_0^{t-u},y^{t-u}\right)\nonumber\\
  & &\;\;\;\;\;\;\times \prod_{\tau=v+1}^t\!\!
    {\rm Pr}\!\left(x_{\tau}\left|x^{\tau-1},s_0^{\tau-u-1},y^{\tau-u-1}\!\right.\right)\nonumber\\
  & &\;\;\;\;\;\;\;\;\;\;\;\;\;\;\;\;\times\,
    {\rm Pr}\!\left(y_{\tau-u},s_{\tau-u}\left|x^{\tau},s_0^{\tau-u-1},y^{\tau-u-1}\!\right.\right)
\end{eqnarray}
The channel laws
${\rm Pr}\!\left(y_{\tau-u},s_{\tau-u}\!\left|x^{\tau},s_0^{\tau-u-1},y^{\tau-u-1}\right.\!\right)$
and
${\rm Pr}\left(\left.y_{t-u+1}^t\right|x^t,s_0^{t-u},y^{t-u}\right)$ in the above equation
are both independent of the source distribution $\mathcal{P}_1$~(or $\mathcal{P}_2$) since
\begin{eqnarray}\label{eqnY1}
    \lefteqn{{\rm Pr}\left(y_{\tau-u},s_{\tau-u}\left|x^{\tau},s_0^{\tau-u-1},y^{\tau-u-1}\right.\right)} \hspace{0.5cm}\nonumber\\
    &\stackrel{\rm (a)}{=}& \!\!\!
        {\rm Pr}\left(y_{\tau-u}\left|x_{\tau-u},s_{\tau-u-1}\right.\right)
        {\rm Pr}\left(s_{\tau-u}\left|s_{\tau-u-1}\right.\right)\nonumber\\
    &\stackrel{\rm (b)}{=}& \!\!\!
        {\rm Pr}\left(y_{\tau-u},s_{\tau-u}\left|x_{\tau-v}^{\tau},s_{\tau-v-1}^{\tau-u-1},y^{\tau-u-1}\right.\right)
\end{eqnarray}
and
\begin{eqnarray}\label{eqnY}
    \lefteqn{{\rm Pr}\left(\left.y_{t-u+1}^t\right|x^t,s_0^{t-u},y^{t-u}\right)} \nonumber\\
    &=& \sum_{s_{t-u+1}^t}{\rm Pr}\left(\left.y_{t-u+1}^t,s_{t-u+1}^t\right|x^t,s_0^{t-u},y^{t-u}\right)\nonumber\\
    &=& \sum_{s_{t-u+1}^t}\prod_{\tau=t-u+1}^t
        {\rm Pr}\left(\left.y_\tau,s_\tau\right|x^t,s_0^{\tau-1},y^{\tau-1}\right)\nonumber\\
    &\stackrel{\rm (c)}{=}& \sum_{s_{t-u+1}^t}\prod_{\tau=t-u+1}^t
        {\rm Pr}\left(\left.y_\tau\right|x_\tau,s_{\tau-1}\right)
        {\rm Pr}\left(\left.s_\tau\right|s_{\tau-1}\right)\nonumber\\
    &\stackrel{\rm (d)}{=}& \sum_{s_{t-u+1}^t}\prod_{\tau=t-u+1}^t
        {\rm Pr}\left(\left.y_\tau,s_\tau\right|x_{t-v}^t,s_{t-v-1}^{\tau-1},y^{\tau-1}\right)\nonumber\\
    &=& {\rm Pr}\left(\left.y_{t-u+1}^t\right|x_{t-v}^t,s_{t-v-1}^{t-u},y^{t-u}\right)
\end{eqnarray}
where equalities (a), (b), (c) and (d) result from
Proposition~\ref{Propnewsource} and the assumption $u\leq v$.
Equalities (a) and (c) also state that the conditional
probabilities
${\rm Pr}\!\left(y_{\tau\!-\!u},s_{\tau\!-\!u}\left|x^{\tau}\!,\!s_0^{\tau\!-\!u\!-\!1},\!y^{\tau\!-\!u-1}\right.\!\right)$
and
${\rm Pr}\!\left(\!\left.y_{t-u+1}^t\right|x^t\!,\!s_0^{t\!-\!u},\!y^{t\!-\!u}\!\right)$
are completely determined by the channel transition law.

Therefore, using~(\ref{eqnY1}) and~(\ref{eqnY}), the given source $\mathcal{P}_1\!\in\!\mathcal{P}(u,u)$
induces the joint probability
\begin{eqnarray}\label{source1}
  \lefteqn{{\rm Pr}^{(\mathcal{P}_1)}\!\!\left(x_{t-v}^t,s_{t-v-1},y^t\right)} \nonumber\\
  &\!\!\!\!\!\!=&\!\!\!\!\!\!\!\!\!\!\! \sum_{x^{t\!-\!v\!-\!1}\!,s_0^{t\!-\!v\!-\!2}\!,s_{t\!-\!v}^{t\!-\!u}}\!\!\!\!\!\!\!\!\!\!\!\!\!
     {\rm Pr}^{(\mathcal{P}_1)}\!\!\left(x^v\!,s_0^{v-u}\!,y^{v-u}\!\right)
     {\rm Pr}\!\left(\left.\!y_{t-u+1}^t \!\right|\!x_{t-v}^t,s_{t-v-1}^{t-u},y^{t-u}\!\right) \nonumber\\
  & &\;\;\;\;\times\! \prod_{\tau={v+1}}^t\!\!
     {\rm Pr}^{(\mathcal{P}_1)}\!\!\left(x_{\tau}\left|x^{\tau-1},s_0^{\tau-u-1},y^{\tau-u-1}\!\right.\right)\nonumber\\
  & &\;\;\;\;\;\;\;\;\;\;\;\;\times\;
    {\rm Pr}\!\left(y_{\tau-u},s_{\tau-u} \!\left|x_{\tau-v}^{\tau},s_{\tau-v-1}^{\tau-u-1},y^{\tau-u-1}\right.\!\right)
\end{eqnarray}
and the conditional probability
\begin{eqnarray}\label{eqnCondProb}
  \lefteqn{{\rm Pr}^{(\mathcal{P}_1)}\!\left(x_t\left|x_{t-v}^{t-1},s_{t-v-1}^{t-u-1},y^{t-u-1}\right.\right)} \nonumber \\
  &=& \frac{{\rm Pr}^{(\mathcal{P}_1)}\!\left(x_{t-v}^t,s_{t-v-1}^{t-u-1},y^{t-u-1}\right)}
    {{\rm Pr}^{(\mathcal{P}_1)}\!\left(x_{t-v}^{t-1},s_{t-v-1}^{t-u-1},y^{t-u-1}\right)}\\
  &=& \frac{\sum\limits_{x^{t-v-1},s_0^{t-v-2}}
    {\rm Pr}^{(\mathcal{P}_1)}\!\left(x^t,s_0^{t-u-1},y^{t-u-1}\right)}
    {\sum\limits_{x^{t-v-1},s_0^{t-v-2}}
    {\rm Pr}^{(\mathcal{P}_1)}\!\left(x^{t-1},s_0^{t-u-1},y^{t-u-1}\right)}
\end{eqnarray}
where
\begin{eqnarray*}
\lefteqn{{\rm Pr}^{(\mathcal{P}_1)}\!\!\left(x^t\!,s_0^{t\!-\!u\!-\!1}\!,y^{t\!-\!u\!-\!1}\right)} \nonumber\\
&\!=&\!\!{\rm Pr}^{(\mathcal{P}_1)}\!\!\left(\!x^{t\!-\!1}\!,\!s_0^{t\!-\!u\!-\!1}\!,\!y^{t\!-\!u\!-\!1}\right)
{\rm Pr}^{(\mathcal{P}_1)}\!\!\left(\!x_t\!\left|x^{t\!-\!1}\!,\!s_0^{t\!-\!u\!-\!1}\!,\!y^{t\!-\!u\!-\!1}\right.\!\right)\!
\end{eqnarray*}
and
\begin{eqnarray*}
\lefteqn{{\rm Pr}^{(\mathcal{P}_1)}\!\left(x^{t-1},s_0^{t-u-1},y^{t-u-1}\right)} \nonumber\\
&=&\!\!\!{\rm Pr}^{(\mathcal{P}_1)}\!\left(x^v,s_0^{v-u},y^{v-u}\right)\nonumber\\
& & \!\!\times\!\prod\limits_{\tau=v+1}^{t-1}\!\!
    {\rm Pr}^{\left(\mathcal{P}_1\right)}\!\left(x_\tau\!\left|x^{\tau-1},s_0^{\tau-u-1},y^{\tau-u-1}\right.\!\right)\nonumber\\
& & \:\:\:\:\:\:\:\:\times\,{\rm Pr}\!\left(y_{\tau-u}\!\left|x_{\tau-u},s_{\tau-u-1}\right.\!\right)
        {\rm Pr}\!\left(s_{\tau-u}\!\left|s_{\tau-u-1}\right.\!\right).
\end{eqnarray*}
On the other hand, the source
$\mathcal{P}_2\in\!\mathcal{P}_v(u,u)$ constructed
as~(\ref{eqnbuiltsource}) induces the joint probability shown in~(\ref{source2Joint})~(see the top of the following page),
\begin{figure*}
\normalsize
\begin{eqnarray}\label{source2Joint}
  \lefteqn{{\rm Pr}^{(\mathcal{P}_2)}\!\left(x_{t-v}^t,s_{t-v-1},y^t\right)} \nonumber\\
  &\!\!\!\!\!\!\!\!=&
    \!\!\!\!\!\!\!\!\!\!\!\! \sum_{x^{t\!-\!v\!-\!1}\!,s_0^{t\!-\!v\!-\!2}\!,s_{t\!-\!v}^{t\!-\!u}} \!\!\!\!\!\!\!\!\!\!\!\!\!
    {\rm Pr}^{(\mathcal{P}_2)}\!\!\left(x^{v}\!,s_0^{v\!-\!u}\!,y^{v\!-\!u}\right)\!
    {\rm Pr}\!\left(\!\left.y_{t\!-\!u\!+\!1}^t\!\right|\!x_{t\!-\!v}^t,s_{t\!-\!v\!-\!1}^{t\!-\!u},y^{t\!-\!u}\right)\!\!\!
    \prod_{\tau={v\!+\!1}}^t \!\!\!{\rm Pr}^{(\mathcal{P}_2)}\!\!\left(x_{\tau}\!\left|x^{\tau\!-\!1}\!,s_0^{\tau\!-\!u\!-\!1}\!,y^{\tau\!-\!u\!-\!1}\!\right.\right)\!
    {\rm Pr}\!\left(y_{\tau\!-\!u},s_{\tau\!-\!u}\!\left|x_{\tau\!-\!v}^{\tau},s_{\tau\!-\!v\!-\!1}^{\tau\!-\!u\!-\!1},y^{\tau\!-\!u\!-\!1}\!\right.\right) \nonumber\\
  &\!\!\!\!\!\!\!\!\stackrel{\textrm{(e)}}{=}&
    \!\!\!\!\!\!\!\!\!\!\!\! \sum_{x^{t\!-\!v\!-\!1}\!,s_0^{t\!-\!v\!-\!2}\!,s_{t\!-\!v}^{t\!-\!u}} \!\!\!\!\!\!\!\!\!\!\!\!\!
    {\rm Pr}^{(\mathcal{P}_1)}\!\!\left(x^{v}\!,s_0^{v\!-\!u}\!,y^{v\!-\!u}\right)\!
    {\rm Pr}\!\left(\!\left.y_{t\!-\!u\!+\!1}^t\!\right|\!x_{t\!-\!v}^t,s_{t\!-\!v\!-\!1}^{t\!-\!u},y^{t\!-\!u}\right)\!\!\!
     \prod_{\tau=v\!+\!1}^t \!\!\! {\rm Pr}^{\left(\mathcal{P}_1\right)}\!\!
     \left(x_\tau\!\left|x_{\tau\!-\!v}^{\tau\!-\!1}\!,s_{\tau\!-\!v\!-\!1}^{\tau\!-\!u\!-\!1}\!,y^{\tau\!-\!u\!-\!1}\!\right.\right)\!
     {\rm Pr}\!\left(y_{\tau\!-\!u},s_{\tau\!-\!u}\!\left|x_{\tau\!-\!v}^{\tau},s_{\tau\!-\!v\!-\!1}^{\tau\!-\!u\!-\!1},y^{\tau\!-\!u\!-\!1}\!\right.\right) \nonumber\\
  &\!\!\!\!\!\!\!\!\stackrel{\textrm{(f)}}{=}&
    \!\!\!\!\!\!\!\!\!\!\!\! \sum_{x^{t\!-\!v\!-\!1}\!,s_0^{t\!-\!v\!-\!2}\!,s_{t\!-\!v}^{t\!-\!u}} \!\!\!\!\!\!\!\!\!\!\!\!\!
    {\rm Pr}^{(\mathcal{P}_1)}\!\!\left(x^{v}\!,s_0^{v\!-\!u}\!,y^{v\!-\!u}\right)\!
    {\rm Pr}\!\left(\left.y_{t\!-\!u\!+\!1}^t\!\right|\!x_{t\!-\!v}^t,s_{t\!-\!v\!-\!1}^{t\!-\!u},y^{t\!-\!u}\right)\!\!\!
     \prod_{\tau=v\!+\!1}^t \!\frac{{\rm Pr}^{(\mathcal{P}_1)}\!\left(x_{\tau-v}^{\tau},s_{\tau-v-1}^{\tau-u},y^{\tau-u}\right)}
     {{\rm Pr}^{(\mathcal{P}_1)}\!\left(x_{\tau-v}^{\tau-1},s_{\tau-v-1}^{\tau-u-1},y^{\tau-u-1}\right)}\nonumber\\
  &\!\!\!\!\!\!\!\!\stackrel{\textrm{(g)}}{=}& \sum_{s_{t-v}^{t-u}}
     {\rm Pr}^{(\mathcal{P}_1)}\!\!\left(x_{t-v}^{t},s_{t-v-1}^{t-u},y^{t-u}\right)
     {\rm Pr}\!\left(\left.y_{t-u+1}^t\right|x_{t-v}^t,s_{t-v-1}^{t-u},y^{t-u}\right)\nonumber\\
  &\!\!\!\!\!\!\!\!=& {\rm Pr}^{(\mathcal{P}_1)}\!\left(x_{t-v}^t,s_{t-v-1},y^t\right)
\end{eqnarray}
\hrulefill
\end{figure*}
where equality (e) follows from the construction of the source $\mathcal{P}_2$, 
equality (f) results from the conditional probability
in~(\ref{eqnCondProb}), and equality (g) is obtained by summing
and canceling the numerators and the denominators in successive
fractions starting at $\tau = v+1$ and considering
${\rm Pr}^{\left(\mathcal{P}_1\!\right)}\!\left(x^{v},s_0^{v-u},y^{v-u}\right)$.

The equality in~(\ref{source2Joint}) implies that the source
$\mathcal{P}_2\!\in\!\mathcal{P}_v(u,u)\!\subseteq\!\mathcal{P}(u,u)$
induces the same information
$I\!\left(\!X_{t-v}^t,\!S_{t-v-1};\!Y_t\!\left|Y^{t-1}\right.\!\right)$
as the source $\mathcal{P}_1\!\in\!\mathcal{P}(u,u)$ does. Since
$\mathcal{P}_1$ is chosen from $\mathcal{P}(u,u)$ arbitrarily, the
supremum $\mathcal{I}^{*}_{FB,SI}(u,v)$ can be taken over the set
of conditional Markov sources $\mathcal{P}_v(u,u)$ instead of over
the set $\mathcal{P}(u,u)$.
\end{IEEEproof}

\section{Proof of Theorem~\ref{Themnewsource2}}\label{proofThemnewsource2}

\begin{IEEEproof}
For convenience, the conditional probabilities
${\rm Pr}\!\left(x_t\!\left|x_{t-v}^{t-1},s_{t-v-1}^{t-u-1},y^{t-u-1}\right.\!\right)$
and
${\rm Pr}\!\left(x_t\!\left|x_{t-v}^{t-1},s_{t-v-1}^{t-u-1},\underline{\alpha}_{t-1}\right.\!\right)$
are both referred to as {\em policies} at time $t$. To prove
Theorem~\ref{Themnewsource2}, we shall show that the vector of the
{\em a posteriori} probabilities $\underline{\alpha}_{t-1}$ can be used to
replace the delayed feedback $y^{t-u-1}$ for the purpose of
determining the optimal policies that achieve the supremum
$\mathcal{I}^*_{FB,SI}(u,v)$.
First, we show that Bellman's principle of optimality~\cite{Bertsekas05,Bertsekas07} holds. For any time instant $T$ in
the interval $[1,N]$, we decompose the information rate
as
\begin{eqnarray}\label{decompose}
 \lefteqn{\!\!\!\!\!\!\!\!\sum_{t=1}^N
 I\!\left(\!X_{t\!-\!v}^{t},S_{t\!-\!v\!-\!1};Y_t\!\left|Y^{t\!-\!1}\right.\!\right)} \nonumber\\
  &\!\!\!\!\!\!\!\!\!\!\!\!\!\!\!\!\!\!=& \!\!\!\!\!\!\!\! \sum_{t=1}^{T-1} I\!\left(\!X_{t\!-\!v}^{t},S_{t\!-\!v\!-\!1};Y_t\!\left|Y^{t\!-\!1}\right.\!\right)\nonumber\\
  &\!\!\!\!\!\!\!\!\!\!\!\!\!\!\!\! & \!\!\!\!\!\!\!\!\!\!\!\! +\!\!\!\!\sum_{y^{T\!-\!u\!-\!1}}\!\!\!
  {\rm Pr}\!\left(y^{T\!-\!u\!-\!1}\!\right)\!\!\left[\sum_{t=T}^N\!
  I\!\left(\!X_{t\!-\!v}^{t},\!S_{t\!-\!v\!-\!1};\!Y_t\!\!\left|y^{T\!-\!u\!-\!1}\!,Y_{T\!-\!u}^{t\!-\!1}\right.\!\right)\!\right]\!\!.
\end{eqnarray}
Similar to~(\ref{sourceG}) in the proof of Theorem~\ref{Themnewsource1}, we have
\begin{eqnarray}\label{firstpart}
  \lefteqn{\!\!\!\!\!\!\!\!{\rm Pr}\!\left(x^{T-1},s^{T-v-2},y^{T-1}\right)} \nonumber\\
  &\!\!\!\!\!\!\!=& \!\!\!\!\!\!\!\! \sum_{s_{T-v-1}^{T-u-1}} \!\! {\rm Pr}\!\left(x^{T-1},s^{T-u-1},y^{T-1}\right)\nonumber\\
  &\!\!\!\!\!\!\!=& \!\!\!\!\!\!\!\! \sum_{s_{T-v-1}^{T-u-1}} \!\!
    {\rm Pr}\!\left(y_{T-u}^{T-1}\!\left|x_{T-u}^{T-1},s_{T-u-1}\right.\!\right)\nonumber\\
  & & \!\!\!\!\!\times \!\prod_{\tau=1}^{T-1} \!
     {\rm Pr}\!\left(x_\tau\!\left|x_{\tau-v}^{\tau-1},s_{\tau-v-1}^{\tau-u-1},y^{\tau-u-1}\right.\!\right) \nonumber\\
  & & \;\;\times \, {\rm Pr}\!\left(y_{\tau-u}\!\left|x_{\tau-u},\,s_{\tau-u-1}\right.\!\right)
     {\rm Pr}\!\left(s_{\tau-u}\!\left|s_{\tau-u-1}\right.\!\right)
\end{eqnarray}
which is independent of policies after time $T$, i.e., independent
of the policies in the set
$\left\{\left.{\rm Pr}\!\left(x_t\left|x_{t-v}^{t-1},s_{t-v-1}^{t-u-1},y^{t-u-1}\right.\right)
\right|T\!\leq\!t\!\leq\! N\!\right\}$. Therefore, if optimal policies from time $1$ to
$N$ are given, then the corresponding policies after time $T$ must
be optimal in the sense that they maximize the last term of~(\ref{decompose}). Thus
we have proved Bellman's principle of optimality~\cite{Bertsekas05,Bertsekas07}.

Next, we show that if after time $T$ we utilize policies
\begin{equation*}
    \left\{\left.{\rm Pr}\!\left(x_t\!\left|x_{t-v}^{t-1},s_{t-v-1}^{t-u-1},
    \underline{\alpha}_{T-1},y_{T-u}^{t-u-1}\right.\!\right)\right|T\!\leq\!t\!\leq\!N\right\}
\end{equation*}
instead of the general policies
\begin{equation*}
\left\{\left.{\rm Pr}\!\left(x_t\!\left|x_{t-v}^{t-1},s_{t-v-1}^{t-u-1},y^{T-u-1},y_{T-u}^{t-u-1}\right.\!\right)\right|T\!\leq\!t\!\leq\!N\right\}
\end{equation*}
we can still maximize the last term in (\ref{decompose}). To show
this, suppose that two different sequences $y^{T-u-1}$ and
$\tilde{y}^{T-u-1}$ induce the same a posteriori probability
vectors $\underline{\alpha}_{T-1}$ and
$\tilde{\underline{\alpha}}_{T-1}$, that is, for all
$\left(x_{T-v}^{T-1},\,s_{T-v-1}^{t-u-1}\!\right)$, we have
\begin{equation*}\label{samealpha}
  \alpha_{T-1}\!\left(x_{T-v}^{T-1},s_{T-v-1}^{T-u-1}\!\right)=\tilde{\alpha}_{T-1}\!\left(x_{T-v}^{T-1},s_{T-v-1}^{T-u-1}\!\right).
\end{equation*}
For the different sequences $y^{T-u-1}$ and $\tilde{y}^{T-u-1}$,
if we use the same policies after time $T$, i.e., for all $t$ in
the interval $T\!\leq\! t\! \leq\! N$,
\begin{eqnarray*}\label{beyond}
  \lefteqn{\hspace{-1.0cm}{\rm Pr}\left(x_t\left|x_{t-v}^{t-1},s_{t-v-1}^{t-u-1},y^{T-u-1},y_{T-u}^{t-u-1}\right.\right)} \nonumber \\
  &=& \!\!{\rm Pr}\left(x_t\left|x_{t-v}^{t-1},s_{t-v-1}^{t-u-1},\tilde{y}^{T-u-1},y_{T-u}^{t-u-1}\right.\right)
\end{eqnarray*}
then we have
\begin{eqnarray}\label{samemeasure}
  \lefteqn{{\rm Pr}\left(x_{T-v}^{N},s_{T-v-1}^{N-v-1},y_{T-u}^N\left|y^{T-u-1}\right.\right)} \nonumber\\
  &\!\!=& \!\!\!\!\!\! \sum_{s_{N-v}^{N-u}}
    {\rm Pr}\left(x_{T-v}^{N},s_{T-v-1}^{N-u},y_{T-u}^N\left|y^{T-u-1}\right.\right)\nonumber\\
  &\!\!=& \!\!\!\! \sum_{s_{N-v}^{N-u}}
    {\rm Pr}\left(x_{T-v}^{T-1},s_{T-v-1}^{T-u-1}\left|y^{T-u-1}\right.\right) \nonumber\\
  &\!\! & \;\times\,{\rm Pr}\left(x_{T}^{N},s_{T-u}^{N-u},y_{T-u}^{N-u}\left|x_{T-v}^{T-1},s_{T-v-1}^{T-u-1},y^{T-u-1}\right.\right)\nonumber\\
  &\!\! &\;\times\, {\rm Pr}\left(y_{N-u+1}^{N}\left|x_{T-v}^{N},s_{T-v-1}^{N-u},y^{N-u}\right.\right)\nonumber\\
  &\!\!\stackrel{\rm (h)}{=}& \!\!\!\! \sum_{s_{N-v}^{N-u}}\!
    \alpha_{T-1}\!\left(x_{T-v}^{T-1},s_{T-v-1}^{T-u-1}\right)
    {\rm Pr}\left(y_{N-u+1}^{N}\left|x_{N-u+1}^{N},s_{N-u}\right.\!\right)\nonumber\\
  &\!\!&\;\times \prod_{\tau=T}^N
    {\rm Pr}\!\left(x_\tau\left|x_{\tau-v}^{\tau-1},s_{\tau-v-1}^{\tau-u-1},y^{T-u-1},
    y_{T-u}^{\tau-u-1}\right.\!\right) \nonumber\\
  &\!\!&\;\;\;\;\;\;\;\;\times\,
    {\rm Pr}\left(y_{\tau-u}|x_{\tau-u},s_{\tau-u-1}\right)
    {\rm Pr}\left(s_{\tau-u}|s_{\tau-u-1}\right)\nonumber\\
  &\!\!=& \!\!\!\! \sum_{s_{N-v}^{N-u}} \!
    \tilde{\alpha}_{T-1}\!\left(x_{T-v}^{T-1},s_{T-v-1}^{T-u-1}\right)
    {\rm Pr}\!\left(y_{N-u+1}^{N}\left|x_{N-u+1}^{N},s_{N-u}\right.\!\right)\nonumber\\
  &\!\!&\;\times \prod_{\tau=T}^N
    {\rm Pr}\left(x_\tau\left|x_{\tau-v}^{\tau-1},s_{\tau-v-1}^{\tau-u-1},\tilde{y}^{T-u-1},
    y_{T-u}^{\tau-u-1}\right.\right) \nonumber\\
  &\!\!&\;\;\;\;\;\;\;\;\times\,
    {\rm Pr}\left(y_{\tau-u}|x_{\tau-u},s_{\tau-u-1}\right)
    {\rm Pr}\left(s_{\tau-u}|s_{\tau-u-1}\right)\nonumber\\
  &\!\!\stackrel{\rm (i)}{=} & \!\!
    {\rm Pr}\left(x_{T-v}^{N},s_{T-v-1}^{N-v-1},y_{T-u}^N\left|\tilde{y}^{T-u-1}\right.\right)
\end{eqnarray}
where equalities (h) and (i) result from Proposition~\ref{Propnewsource} and the assumption $u\leq v$.
The equality in~(\ref{samemeasure}) implies
\begin{eqnarray}\label{beyondrate}
    \lefteqn{\!\!\!\!\!\!\!\!\sum_{t=T}^NI\left(X_{t-v}^{t},S_{t-v-1};Y_t\left|y^{T-u-1},Y_{T-u}^{t-1}\right.\right)} \nonumber \\
    &=& \!\!\!\!\sum\limits_{t=T}^N I\left(X_{t-v}^{t},S_{t-v-1};Y_t\left|\tilde{y}^{T-u-1},Y_{T-u}^{t-1}\right.\right).
\end{eqnarray}
Therefore, the optimal policies after time $T$ for $y^{T-u-1}$
must also be optimal for $\tilde{y}^{T-u-1}$, and vice versa.
Since $y^{T-u-1}$ and $\tilde{y}^{T-u-1}$ induce the same vector
$\underline{\alpha}_{T-1}\!=\! \tilde{\underline{\alpha}}_{T-1}$,
the vector $\underline{\alpha}_{T-1}$ can be used instead of $y^{T-u-1}$, and the
optimal policies after time $T$ can be replaced by
\begin{equation*}
 \left\{\left.{\rm Pr}\!\left(x_t\!\left|x_{t-v}^{t-1},s_{t-v-1}^{t-u-1},\underline{\alpha}_{T-1},y_{T-u}^{t-u-1}\right.\!\right)\!\right|T\!\leq\!t\!\leq\!N\right\}.
\end{equation*}

Since $T$ is chosen arbitrarily, the optimal source in the set
$\mathcal{P}'_v(u,u)\!=\!\left\{{\rm Pr}\!\left(x_t\!\left|x_{t-v}^{t-1},s_{t-v-1}^{t-u-1},\underline{\alpha}_{t-1}\right.\!\right)\right\}_{t=1}^{\infty}$
achieves the same supremum $\mathcal{I}^{*}_{FB,SI}(u,v)$ as the
optimal source in the set $\mathcal{P}_v(u,u)$ does.
\end{IEEEproof}

\section{Proof of Theorem~\ref{ThemDP}}\label{proofThemDP}

\begin{IEEEproof}
Let $\beta \in (0,1)$. We introduce the $\beta$-discounted version of {\bf Problem B},
for all $\underline{\alpha}_{0} \in \hat{\mathcal{A}}$,
\begin{equation}\label{eqnDC}
    \mathcal{I}_{\beta}(\underline{\alpha}_0)
    \!=\! \sup\liminf_{N \rightarrow \infty} \mathbf{E}\!
       \left[\sum\limits_{t=1}^{N} \beta^{t-1} g\!\left(\underline{\alpha}_{t-1}, {\rm p}(\underline{\alpha}_{t-1}), Y_{t-u}\right)\!\right]
\end{equation}
where only stationary policy sequences $\{{\rm p}_t\}_{t=1}^{\infty}$ with
${\rm p}_t = {\rm p} \stackrel{\Delta}{=} \{{\rm p}(\underline{\alpha}): \underline{\alpha} \in \hat{\mathcal{A}}\}$ are considered.
By Proposition~4.1.3 in~\cite{Bertsekas07}, there exists a {\em Blackwell optimal policy} ${\rm p}^* = \{{\rm p}^*(\underline{\alpha}): \underline{\alpha} \in \hat{\mathcal{A}}\}$ that is stationary and
simultaneously optimal for all $\beta$-discounted problems~(\ref{eqnDC}) where $\beta$ is sufficiently close to $1$. From Proposition~4.1.7 in~\cite{Bertsekas07}, we know that the Blackwell optimal policy ${\rm p}^*$ is optimal over all policies for {\bf Problem B}.~(These results can also be obtained according to Theorem~4.3 in~\cite{Arapostathis93}).
\end{IEEEproof}

\section*{Acknowledgment}

The authors would like to thank Dr. Shaohua Yang for his helpful advice at the beginning of this work, and Prof. Xianping Guo for providing helpful references on Markov decision processes. The authors are also grateful to reviewers for their helpful comments, who also pointed out some errors in the previous versions of the paper.

\ifCLASSOPTIONcaptionsoff
  \newpage
\fi



\bibliographystyle{IEEEtran}
\bibliography{IEEEabrv,mybibfile}

\begin{IEEEbiographynophoto}{Xiujie~Huang}(S'10)
received the M.Sc. degree in mathematics from Sun Yat-sen University, Guangzhou, China, in 2006.
She is a Ph.D. candidate in the Department of Electronics and Communication Engineering,
Sun Yat-sen University, Guangzhou, China.
She is currently visiting the Department of Electrical Engineering, University of Hawaii, Honolulu, USA.

Her research interests include information theory and its applications in digital communication and storage systems.
Her current researches focus on capacity~(region) computation and magnetic recording/flash memory channel medeling.

\end{IEEEbiographynophoto}

\begin{IEEEbiographynophoto}{Aleksandar~Kav\v{c}i\'{c}}(S'93-M'98-SM'04)
received the Dipl. Ing. degree in electrical engineering from
Ruhr-University, Bochum, Germany, in 1993, and the Ph.D. degree in electrical
and computer engineering from Carnegie Mellon University, Pittsburgh, PA, in 1998.

Since 2007, he has been with the University of Hawaii, Honolulu, where he
is presently a Professor of Electrical Engineering. Prior to 2007, he
was with the Division of Engineering and Applied Sciences, Harvard University,
Cambridge, MA. He served as a Visiting
Associate Professor with the City University of Hong Kong in Fall 2005 and as
a Visiting Scholar with the Chinese University of Hong Kong in Spring 2006.
Prof. Kav\v{c}i\'{c} received the IBM Partnership Award in 1999 and the NSF
CAREER Award in 2000. He is a corecipient, with X.~Ma and N.~Varnica, of
the 2005 IEEE Best Paper Award in Signal Processing and Coding for Data
Storage. He served on the Editorial Board of the IEEE TRANSACTIONS ON
INFORMATION THEORY as Associate Editor for Detection and Estimation from
2001 to 2004, as Guest Editor of the IEEE SIGNAL PROCESSING MAGAZINE
during 2003-004, and as Guest Editor of the IEEE JOURNAL ON SELECTED
AREAS IN COMMUNICATIONS from 2008 to 2009. From 2005 until 2007, he was
the Chair of the Data Storage Technical Committee of the IEEE Communications
Society.
\end{IEEEbiographynophoto}

\begin{IEEEbiographynophoto}{Xiao~Ma}(M'08)
received the Ph.D. degree in communication and information systems
from Xidian University, China, in 2000.
From 2000 to 2002, he was a Postdoctoral Fellow with Harvard University,
Cambridge, MA. From 2002 to 2004, he was a Research Fellow with City University
of Hong Kong. He is now a Professor with the Department of Electronics
and Communication Engineering, Sun Yat-sen University, Guangzhou, China.

His research interests include information theory, channel coding theory and
their applications to communication systems and digital recording systems.

Dr. Ma is a corecipient, with A.~Kav\v{c}i\'{c} and N.~Varnica, of the 2005 IEEE
Best Paper Award in Signal Processing and Coding for Data Storage.
In 2006, Dr. Ma received the Microsoft Professorship Award from Microsoft Research Asia.
\end{IEEEbiographynophoto}




\end{document}